\documentclass[a4paper, 11pt]{article} 
\usepackage[margin=2.4cm, top=2.5cm, bottom=5cm, footskip=3cm]{geometry}	
\usepackage{amsmath}		
\usepackage{amsfonts}		
\usepackage{amsthm}
\usepackage[utf8]{inputenc}
\usepackage{color}
\usepackage{tikz}
\usepackage{slashed}
\usetikzlibrary{calc}
\usetikzlibrary{arrows}  
\usetikzlibrary{decorations.markings}
\usetikzlibrary{shapes,snakes}
\usepackage{hyperref} 
\hypersetup{
	bookmarks=true,
	unicode=false,
	pdftoolbar=true,
	pdfmenubar=true,
	pdffitwindow=false,
	pdfstartview={FitH},
	pdftitle={My title},
	pdfauthor={Author},
	pdfsubject={Subject},
	pdfcreator={Creator},
	pdfproducer={Producer},
	pdfkeywords={keyword1} {key2} {key3},
	pdfnewwindow=true,
	colorlinks=true,   
	linkcolor={blue!80!black}, 
	citecolor={green!50!black}, 
	filecolor={blue!80!black},
	urlcolor={cyan},      
} 
\usepackage[figure]{hypcap}
\usepackage{amssymb}		
\usepackage{amstext}		
\usepackage{graphicx}		
\usepackage[utf8]{inputenc}	
\usepackage[T1]{fontenc}	
\usepackage{pdfpages}		
\usepackage{wrapfig}		
\usepackage{here}			
\usepackage{float}			
\usepackage{fancyhdr}
\usepackage{verbatim}		
\usepackage{wasysym}		
\usepackage{tocloft}		
\usepackage{siunitx}		
\usepackage[arrow, matrix, curve]{xy} 
\usepackage{wrapfig,tabularx,multirow,hyperref,icomma}
\usepackage{mathptmx}
\usepackage{lmodern}
\usepackage{courier}
\usepackage{paralist}
\usepackage{afterpage}
\usetikzlibrary{arrows.meta}
\usepackage{subcaption}
\usetikzlibrary{decorations.pathmorphing,shapes}
\usepackage[section]{placeins}
{\left.\begin{aligned}}
	{\end{aligned}\right\rbrace}
\usepackage{titletoc}
\graphicspath{{figures/}} 

\usepackage{relsize,exscale}

\definecolor{RoyalBlue}{RGB}{35,32, 57}
\definecolor{lightblue}{RGB}{132,131,143}

\tikzstyle{chapterbox} = [fill=black!25!white, thick,
rectangle, inner sep=10pt, inner ysep=5pt]

\tikzstyle{bluebox} = [draw=blue!50!black, fill=blue!10, very thick,
rectangle, inner sep=10pt, inner ysep=5pt] 

\tikzstyle{rahmen blau} = [draw=black, fill=blue!20, very thick,
rectangle, inner sep=10pt, inner ysep=5pt] 

\tikzstyle{graybox} = [draw=black, fill=gray!10, thick,
rectangle, inner sep=10pt, inner ysep=5pt] 

\tikzstyle{box} = [draw=black,  thick,
rectangle, inner sep=10pt, inner ysep=5pt] 

\theoremstyle{definition}
\newtheorem{definition}{Definition}[section]

\newtheorem{theorem}[definition]{Theorem}

\newtheorem{corollary}[definition]{Corollary}

\newtheorem{lemma}[definition]{Lemma}

\theoremstyle{definition}
\newtheorem{remark}[definition]{Remark}

\theoremstyle{definition}

\newcommand{\mylog}{\overset{\swarrow}{\log}} 
\newcommand{\mmysqrt}{\sqrt[\rightarrow]} 

\renewcommand{\Im}{\mathrm{Im}} 
\renewcommand{\Re}{\mathrm{Re}} 

\newcommand{\x}{\boldsymbol{x}}

\newcommand{\go}{\mathrm{go}}
\newcommand{\diff}{\mathrm{diff}}
\newcommand{\iin}{\mathrm{in}}
\newcommand{\ssc}{\mathrm{sc}}
\renewcommand{\iint}{\int\hspace{-.2cm}\int}

\newcommand{\D}{\mathcal{D}}
\renewcommand{\S}{\mathcal{S}}	
\newcommand{\LHP}{\mathrm{LHP}}
\newcommand{\UHP}{\mathrm{UHP}}

\newcommand\numberthis{\addtocounter{equation}{1}\tag{\theequation}} 
\numberwithin{equation}{section} 

\title{Diffraction by a Right-Angled No-Contrast Penetrable Wedge Revisited: A Double Wiener-Hopf Approach}
\date{\hfill }
\author{Valentin D. Kunz\thanks{valentin.kunz@manchester.ac.uk} \ and  \ Raphael C. Assier\thanks{raphael.assier@manchester.ac.uk}  \\ \footnotesize The University of Manchester, Department of Mathematics,  Oxford Road, Manchester, M13 9PL, UK}

\newcommand{\RED}{} 

\begin{document}

\maketitle


\begin{abstract}
	In this paper, we revisit Radlow's innovative approach to diffraction by a penetrable wedge by means of a double Wiener-Hopf technique. We provide a constructive way of obtaining his ansatz and give yet another reason for why his ansatz cannot be the true solution to the diffraction problem at hand. The two-complex-variable Wiener-Hopf equation is reduced to a system of two equations, one of which contains Radlow's ansatz plus some correction term consisting of an explicitly known integral operator applied to a yet unknown function,  whereas the other equation, the compatibility equation, governs the behaviour of this unknown function. 
\end{abstract}

\section{Introduction}
Although diffraction is a well-known phenomenon with a rigorous mathematical theory that emerged in the late 19th century (Sommerfeld, Poincar\'e), many canonical problems still remain unsolved in the sense that no clear analytical solution has been found for them. Here `canonical' refers to problems where the scatterer's geometry is simple but possibly exhibits some singularities making the seemingly easy scattering problem challenging to solve. One of these unsolved problems is the diffraction by a penetrable wedge, that is by a wedge-shaped scatterer made of a material with acoustic (or electromagnetic) properties different from those of the ambient medium. Wedge diffraction problems are of great importance to mathematical, physical, and engineering sciences as they represent one of the building blocks of the geometrical theory of diffraction (GTD, \cite{Keller}). For example, gaining a better understanding of penetrable wedge diffraction is expected to improve numerical methods for high frequency penetrable convex polygon diffraction, see \cite{GrothEtAl.1, GrothEtAl.2}, and has applications in the scattering of light by atmospheric particles such as ice crystals \cite{Baran} which directly feeds into climate change models \cite{SmithEtAl.}.  \\

\emph{Wedge diffraction problems: an overview.} \ Soon after providing his solution to the half-plane problem \cite{Sommerfeld1, Sommerfeld2}, Sommerfeld managed to solve the more general problem of diffraction by a non-penetrable wedge with opening angle $q\pi, \ q \in \mathbb{Q}$ in 1901 (c.f.\! \cite{Sommerfeld3} end of Chapter 5) using the method of Sommerfeld surfaces (see also \cite{AssierShanin3} for an overview and more recent applications of this method). Unfortunately, it has thus far not been possible to generalise this technique to the penetrable wedge and during the past century, new methods have been developed, not only for penetrable wedges but for diffraction problems in general. In particular, many methods including the Sommerfeld-Malyuzhinets method (c.f.\! \cite{Babich}, \cite{Malyuzhinets}), the Wiener-Hopf technique (c.f.\! \RED{\cite{BelinskiyEtAl1973}}, \cite{LawrieAbrahams}, \cite{Noble}), and the Kontorovich-Lebedev
transform approach (c.f.\! \cite{KontorovichLebedev}) have been developed for non-penetrable wedge diffraction; we refer to \cite{Nethercote1} for a review of these and other methods. Moreover, some of these methods have been helpful in gaining a better understanding of penetrable wedge diffraction. Indeed, 
in 2011, Daniele and Lombardi \cite{DanieleLombardi} employed the Wiener-Hopf technique for the isotropic penetrable wedge problem to obtain a system of four Fredholm integral equations which is then solved numerically using quadrature schemes.

Other innovative approaches suitable for high contrast penetrable wedge problems were given by Lyalinov \cite{Lyalinov} and, more recently, by Nethercote et al. \cite{Nethercote2}. In \cite{Lyalinov} Lyalinov uses the Sommerfeld–Malyuzhinets technique to obtain a system of two coupled Malyuzhinets equations which were solved approximately, giving the leading order far-field behaviour, whereas in \cite{Nethercote2} Nethercote et al. combine the Wiener-Hopf and Sommerfeld-Malyuzhinets method. In \cite{Nethercote2} a solution to the penetrable wedge is given as an infinite series of impenetrable wedge problems. Each of these impenetrable wedge problems is   solved exactly, and the resulting infinite series for the penetrable wedge can be evaluated rapidly and efficiently using asymptotic and numerical methods. 

It is important to note that there have been many other approaches as well (this list is, however, by no means exhaustive):
In 1998, Budaev and Bogy looked at finding the pressure field of acoustic wave diffraction by two penetrable wedges 
using the Sommerfeld-Malyuzhinets technique, resulting in a system of eight singular integral equations of Fredholm type (see \cite{Budaev1}) which is solved in \cite{Budaev2} by Neumann series assuming the contrast is close to unity and the wedge's opening angle is small. The following year, Rawlins considered the case of similar wave numbers $k_1 \approx k_2 $ in the electromagnetic setting (dielectric wedge) and  used the Kontorovich-Lebedev transform to create a system of Fredholm integral equations which were solved iteratively to obtain a first order approximations of the diffracted field \cite{Rawlins}.

There have also been some approaches using simple layer potential theory, as discussed in \cite{CroisilleLebeau}, which were employed in 2008 by Babich and Mokeeva. In \cite{BabichMokeeva} they showed that the problem of diffraction by a penetrable wedge has a unique solution and later, in 2012, developed a numerical solution of those simple layer potentials \cite{BabichEtAl.}. \\

\emph{Complex analysis in several variables: a new ansatz.} Whenever any of the previously mentioned methods employed complex analysis, these were one dimensional techniques. Surprisingly, using two dimensional complex analysis, there seems to be a rather straight-forward method of getting a Wiener-Hopf equation. As in his work for the 3D diffraction by a quarter-plane \cite{Radlow1}, Radlow obtained a Wiener-Hopf equation in two complex variables for the 2D right-angled no-contrast penetrable wedge \cite{Radlow2}. His simple yet innovative idea was to use two dimensional Laplace transforms to integrate over the scatterer and thus `capture' the scatterer's geometry as was already done for the half-plane problem using the one dimensional Laplace transform \cite{Noble}. These functional equations are perfectly valid, however, the closed form solutions thus found by Radlow for the quarter-plane and penetrable wedge diffraction problems turned out to be erroneous as they led to the wrong type of near-field behaviour \cite{Meister, KrautLehmann}. Nonetheless, solving these functional equations in several complex variables would be a tremendous achievement in diffraction theory. Unfortunately, the solutions in \cite{Radlow1} and \cite{Radlow2} were not given constructively making it difficult to understand the reasoning behind Radlow's work and pinpoint where he went wrong. Recently, Assier and Abrahams \cite{AssierAbrahams1} revisited Radlow's approach, giving a constructive procedure to obtain his quarter-plane ansatz \emph{plus some correction term}, while Assier and Shanin studied the analyticity properties of the unknowns of Radlow's quarter-plane Wiener-Hopf equation \cite{AssierShanin}. Although the Wiener-Hopf equation remains unsolved, the simpler problem of diffraction by a quarter-plane without incident wave (i.e. a source located at the quarter-plane's tip) has been solved using these novel complex analysis methods \cite{AssierShanin2}, confirming their usefulness. In the present work we will show that the method of \cite{AssierAbrahams1} can also be applied to the right-angled no-contrast penetrable wedge problem. \\


\emph{Aim and plan of the article.}  In the present work
we revisit Radlow's approach \cite{Radlow2} in the spirit of \cite{AssierAbrahams1}. In particular, we provide a constructive procedure of obtaining Radlow's solution \emph{plus some correction term.}  This provides yet another way of showing that Radlow's solution was erroneous. However, the extra term contains a complicated integral operator applied to a yet unknown function. Fortunately, we also obtain a \emph{compatibility equation} involving only the extra term's unknown function and although it has thus far not been possible to solve this equation exactly (which would then provide a closed form solution to the diffraction problem), we strongly believe that it can be employed to accurately test approximations, c.f.\! \cite{AssierAbrahams2}, which will be the subject of future work. We will focus on the special case of a right-angled no-contrast penetrable wedge. 

In Section \ref{Section2}, the diffraction problem is formulated in physical space and thereafter transformed into (two complex dimensional) Fourier space resulting in the problem's Wiener-Hopf equation. In Section \ref{Section3}, the machinery required to work with this equation is introduced and the functional equation's kernel is factorised; we will employ the method of phase portraits (see \cite{Wegert}) to visualise functions of complex variables, which often provides a visually convincing method of verifying results that are tedious to prove otherwise. In Section \ref{Section4} we apply the factorisation techniques developed by Assier and Abrahams in \cite{AssierAbrahams1} to derive a set of equations linking the unknowns of the functional problem. The first equation involves Radlow's solution plus some additional correction term while the second equation, the compatibility equation, may provide a way of finding/approximating this unknown correction term. Finally, we compare our results with the ones found for the quarter-plane problem.


\section{Wiener Hopf equation for the penetrable wedge}\label{Section2}

\subsection{Problem formulation}\label{sec:ProblemFormulation}
We are considering the problem of diffraction of a plane wave \RED{$\phi_{\iin}$} incident on an infinite, right-angled, penetrable wedge (PW) given by
\[
\text{PW} = \{(x_1,x_2) \in \mathbb{R}^2| \ x_1 \geq 0, x_2 \geq 0 \},
\]
see figure \ref{fig:transparentWedge}. We assume transparency of the wedge and thus expect a scattered field \RED{$\phi_{\ssc}$} in $\mathbb{R}^2 \setminus \text{PW}$ and a transmitted field \RED{$\psi$} in PW (c.f.\! figure \ref{fig:transparentWedge}, left).
As usual in scattering problems we assume time-harmonicity with the $e^{-i\omega t}$ convention, where $\omega$ is the (angular) frequency. The time dependence is henceforth suppressed and the wave-fields'  dynamics are \RED{therefore} described by \RED{two} Helmholtz equation\RED{s}. \RED{Moreover, suppressing time harmonicity, the incident plane wave \RED{(only supported within $\mathbb{R}^2 \setminus \text{PW}$)} is given by 
	\[
		\phi_{\iin}(\x) = e^{i \boldsymbol{k}_1 \cdot \x},
\]}where $\boldsymbol{k}_1\RED{\in \mathbb{R}^2}$ is the wave vector and  $\x = (x_1,x_2)\RED{\in \mathbb{R}^2}$ (this notation will be used throughout the article).
\RED{Let us focus on the acoustic setting  of sound propagation through a fluid or gas (the electromagnetic setting is briefly discussed in Remark \ref{rem:EM}). Then the field $\phi(\x)$ given by 
\[
\phi(\x) = \phi_{\ssc}(\x) + \phi_{\iin}(\x)
\]
represents the total pressure field in $\mathbb{R}^2\setminus\text{PW}$, $\psi$ represents the total pressure field in $\text{PW}$, and the wave vector  $\boldsymbol{k}_1$ satisfies $|\boldsymbol{k}_1| = k_1$ for the wave number $k_1 = \omega/c_1$, where $c_1$ is the speed of sound relative to the medium  in $\mathbb{R}^2 \setminus\text{PW}$.
}

\RED{Crucial to the present work is that we are describing a \emph{no-contrast penetrable wedge} meaning that the density $\rho_1$ of the medium in $\mathbb{R}^2 \setminus \text{PW}$ (at rest) is the same as density $\rho_2$ of the medium in $\text{PW}$ (at rest). In particular, the \emph{contrast parameter} $\lambda$ which is given by 
\[
\lambda = \frac{\rho_1}{\rho_2}
\]
satisfies
\[
\lambda = 1.
\]	
However, the wave numbers $k_1=\omega/c_1$ and $k_2=\omega/c_2$ inside and outside PW respectively are different even though $\rho_1 = \rho_2$, since the other media properties, the bulk moduli (c.f. \cite{KinslerEtAl}), defining the speeds of sound $c_1$ and $c_2$ are assumed to be different.}

The boundary value problem at hand is then described by equations \eqref{eq:1.1}--\eqref{eq:1.6} below.
\begin{figure}[t]
	\centering
	\includegraphics[width=\textwidth]{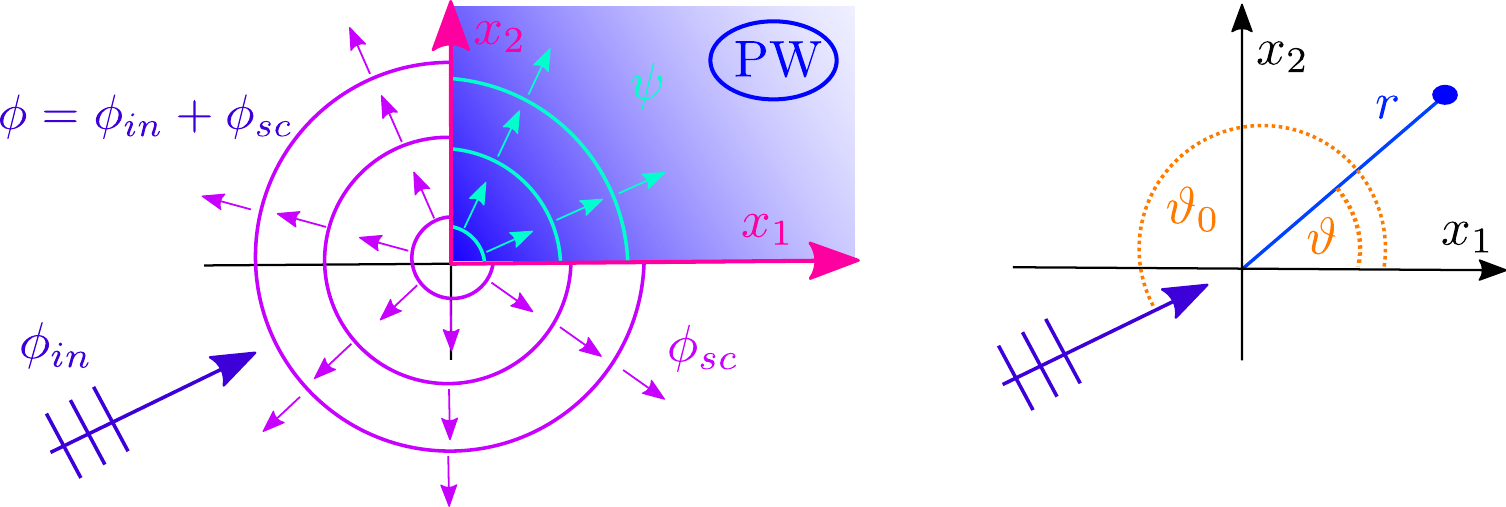}
	\vspace{-.25cm}
	\caption{\small Left: Illustration of the problem described by equations \eqref{eq:1.1}--\eqref{eq:1.6}. The scatterer i.e the penetrable wedge is shown in blue with edges in magenta. Right: Polar coordinate system and incident angle $\vartheta_0$ of $\phi_{\iin}$.} 
	\label{fig:transparentWedge}
\end{figure}
\begin{align}
	&	\Delta \phi + k^2_1 \phi =0 \ \text{in} \  \mathbb{R}^2 \setminus \text{PW}, \label{eq:1.1} \\
	&	\Delta \psi + k^2_2 \psi =0 \ \text{in} \  \text{PW}, \label{eq:1.2} \\[1em]
	&	\phi(0^-, x_2>0) \ \ \ \ = \psi(0^+, x_2>0), \label{eq:1.3}\\
	&   \phi(x_1>0, 0^-) \ \ \ \ = \psi(x_1>0, 0^+), \label{eq:1.4}\\
	&  \partial_{x_1} \phi(0^-, x_2>0) = \partial_{x_1} \psi(0^+, x_2>0), \label{eq:1.5} \\
	&  \partial_{x_2} \phi(x_1>0, 0^-) =  \partial_{x_2} \psi(x_1>0, 0^+).\label{eq:1.6}	
\end{align}
Equations \eqref{eq:1.1} and \eqref{eq:1.2} are the problem's governing equations, describing the fields' dynamics, whereas the boundary conditions \eqref{eq:1.3}--\eqref{eq:1.6} impose continuity of the fields and their normal derivatives at the wedge's boundary. 
\RED{ 
	
	\begin{remark}[The electromagnetic setting]\label{rem:EM}
		Equations \eqref{eq:1.3}--\eqref{eq:1.6} also model the diffraction of an $\boldsymbol{E}$-polarised (resp.\! $\boldsymbol{H}$-polarised) electromagnetic wave incident on a right-angled no-contrast penetrable wedge, where $\boldsymbol{E}$ is the electric field (resp.\! $\boldsymbol{H}$ is the magnetic field). Here $\phi$ corresponds to the total $\boldsymbol{E}$ (resp.\! $\boldsymbol{H}$) field in $\mathbb{R}^2 \setminus \text{PW}$ whereas $\psi$ corresponds to the total $\boldsymbol{E}$ (resp.\! $\boldsymbol{H}$) field in $\text{PW}$ (c.f. \cite{KrautLehmann}, \cite{Nethercote2}, and \cite{Radlow2}). Here, when describing the diffraction of the polarised electric (resp.\! magnetic) field, the assumption that the contrast parameter $\lambda$ satisfies $\lambda=1$ means that the magnetic permeabilities $\mu_{1}$ and $\mu_2$ (resp.\! electric permittivities $\epsilon_{1}$ and $\epsilon_2$) of the medium in $\mathbb{R}^2 \setminus \text{PW}$ and PW respectively satisfy $\mu_1=\mu_2$ (resp.\! $\epsilon_1 = \epsilon_2$). Since in the electromagnetic setting, the wave numbers are given by $k_j = \omega \sqrt{\mu_j \epsilon_j}$ 
		we must have $\epsilon_1 \neq \epsilon_2$ (resp.\! $\mu_1 \neq \mu_2$) for the wave numbers $k_1$ and $k_2$ to be different.
	\end{remark}
}
Now, introducing polar coordinates $(r,\vartheta)$ (c.f.\!  figure \ref{fig:transparentWedge}, right) we can rewrite the incident wave vector $\boldsymbol{k}_1 = -k_1(\cos(\vartheta_0),\sin(\vartheta_0))$ where $\vartheta_0$ is the incident angle. The incident wave can then be rewritten as
\begin{align*}
	\phi_{\iin} = e^{-i(\RED{\mathfrak{a}_1}x_1 + \RED{\mathfrak{a}_2}x_2)} \numberthis \label{eq:IncidentRewritten}
\end{align*}

with \begin{align}
	\RED{\mathfrak{a}_1} = k_1 \cos(\vartheta_0) \RED{\text{ and }} \RED{\mathfrak{a}_2} = k_1 \sin(\vartheta_0). \label{eq.a12Def}
\end{align}
 Henceforth, we assume \RED{$\text{Im}(k_{1}) >0$ and $\Im(k_2) >0$}.  Later on, this condition may be waived by considering the \RED{limits $\Im(k_{1,2}) \to 0$} (see Section \ref{sec:VanishingImPart}). Moreover, we restrict $\vartheta_0 \in (\pi, \frac{3 \pi}{2})$ and $\Re(k_{1,2}) > 0$, so $\Im(\RED{\mathfrak{a}_{1,2}}) < 0$ and $\Re(\RED{\mathfrak{a}_{1,2}})<0$. \RED{Since, as mentioned in the beginning of Section \ref{sec:ProblemFormulation}, we have assumed time harmonicity with the $e^{-i\omega t}$ convention, this corresponds to the damping/absorption of waves.}

\begin{remark}[the general case]
	The situation is more complicated if we allow other incident angles $\vartheta_0$ since then the sign of $\Im(\RED{\mathfrak{a}_{1}})$ and/or $\Im(\RED{\mathfrak{a}_2})$ changes. This technical difficulty can be dealt with by viewing $\RED{\mathfrak{a}_{1,2}}$ as  independent parameters and impose	\RED{$-\Im(\RED{\mathfrak{a}_{1,2}}) >0$}, i.e. give $\RED{\mathfrak{a}_{1,2}}$ an artificial negative imaginary part \emph{irrespective of incident angle and wave number}. Again, once the solution has been obtained, we may take the limit \RED{$\Im(\RED{\mathfrak{a}_{1,2}}) \to 0$}. 
\end{remark}

Finally, 
it is necessary to impose Meixner conditions on the field, ensuring boundedness of energy near the wedge's tip $\x=(0,0)$. That is, for arbitrarily small $\varepsilon > 0$, the following energy integrals need to be finite:
\begin{align}
	&\int_{0}^{\pi/2} \int_{0}^{\varepsilon} r \left(|\nabla \psi|^2 + |\psi|^2\right)dr d \vartheta < \infty,  \label{eq.Edge1} \\
	&\int_{\pi/2}^{2 \pi} \int_{0}^{\varepsilon} r \left(|\nabla \phi|^2 + |\phi|^2\right)dr d \vartheta < \infty. \label{eq.Edge2}
\end{align}	
\RED{Now,} approximating the Helmholtz equation by Laplace's equation near the tip and proposing a separation of variables ansatz \RED{yields a power series expression $\phi_{\ssc} = \sum_{n=1}^{\infty} \left(A_{\nu_n} \sin (\nu_n \vartheta)+B_{\nu_n} \cos (\nu_n \vartheta)\right)r^{\nu_n}$ for $\phi_{\ssc}$ and similarly for $\psi$ near the tip. Then, using \eqref{eq.Edge1}--\eqref{eq.Edge2}}  and the boundary conditions \eqref{eq:1.3}--\eqref{eq:1.6} we find 
\begin{align}
	&	\phi(r, \vartheta) = B +\left(A_1\sin(\vartheta) + B_1\cos(\vartheta)\right)r+ \mathcal{O}(r^2), 
	\ \text{as} \ r \to 0, \label{eq.2.56} \\	
	&	\psi(r, \vartheta) = B + \left(A'_1\sin(\vartheta) + B'_1 \cos(\vartheta)\right)r + \mathcal{O}(r^2), 
	\ \text{as} \ r \to 0, \label{eq.2.57}
\end{align} 
where the constants $B, A_1, A'_1, B_1$ and $B'_1$ are unknown. \RED{We refer to \cite{Babich} and \cite{Jones1964} for a more detailed discussion of this procedure.}	Equations \eqref{eq.2.56} and \eqref{eq.2.57} are the sought edge conditions. It should be noted that \RED{these} particular expressions \RED{ \eqref{eq.2.56} and \eqref{eq.2.57} are} only valid since we have chosen $\lambda =1$. 

\RED{The case of general $\lambda$ has, for instance, been considered in \cite{BabichMokeeva},  \cite{Nethercote2}, and \cite{Rawlins1977}. In \cite{BabichMokeeva} and \cite{Rawlins1977} the behaviour is given up to second order  $\phi = B + \mathcal{O}(r^{\mu})$ (similarly for $\psi$) where the exact value of $\mu >0$ is not specified, whereas in \cite{Nethercote2}, the behaviour is given up to fourth order and an explicit dispersion relation is given for determining $\mu$ (which depends on the contrast parameter $\lambda$)}. 

\begin{remark} 
	Following Radlow's ansatz we would also get $\phi, \psi \sim C \ \text{as} \ r \to 0$ for some suitable constant $C$, see \cite{Radlow2}. However, as pointed out by Kraut and Lehmann in \cite{KrautLehmann}, Radlow's ansatz leads to the wrong value i.e.\! $C \neq B$. \RED{Moreover, in \cite{Rawlins1977} Rawlins explicitly computed the value of $B$ up to second order in $k^2_1 - k^2_2$ when $k_1$ is close to $k_2$, thereby extending Kraut and Lehmann's work.}
\end{remark}

In general, in order for the problem to be well posed, the field also needs to satisfy a radiation condition: The scattered field should be outgoing in the far-field. That is, there are no sources other than the incident wave at infinity. Due to the wavenumbers' positive imaginary part, this is automatically satisfied and by the limiting absorption principle, the radiation condition also holds in the limit $\Im(k_{1,2}) \to 0$. See \RED{\cite{BabichMokeeva},} \cite{Nethercote2} for more information on the radiation condition for penetrable wedges. 

\RED{To conclude this section, we note that specifying the behaviour of the fields near the wedge's tip and at infinity is required to guarantee uniqueness of the solution to the problem described by equations \eqref{eq:1.1}--\eqref{eq:1.6}, see \cite{BabichMokeeva}.}

\subsection{Transformation in Fourier space} 

In this section, the boundary value problem described by \eqref{eq:1.1}--\eqref{eq:1.6} is transformed into Fourier space and the corresponding functional equation is found. Let  $Q_n, \ n=1,2,3,4$ denote the $n$th quadrant of the $(x_1,x_2)$ plane given by
\begin{align*}
	\text{PW} =	& Q_1 = \{ \x \in \mathbb{R}^2| x_1\geq0, \ x_2\geq0 \}, \	Q_2 = \{ \x \in \mathbb{R}^2| x_1\leq0, \ x_2\geq0 \}, \\
	& 	Q_3 = \{ \x \in \mathbb{R}^2| x_1\leq0, \ x_2\leq0 \}, \ Q_4 = \{ \x \in \mathbb{R}^2| x_1\geq0, \ x_2\leq0 \}.
\end{align*}
To derive the problem's functional equation and to keep consistency with recent work on several complex variable methods applied to diffraction problems (c.f.\! \cite{AssierAbrahams1, AssierShanin}) we define:

\begin{definition}[One-quarter Fourier Transform]\label{def. 1/4FT}
	The one-quarter Fourier transform of a function $u$ is given by 
	\begin{align}
		U_{1/4}(\boldsymbol{\alpha})=\mathcal{F}_{1/4}[u](\boldsymbol{\alpha}) = \iint_{Q_1} u(\x) e^{i \boldsymbol{\alpha} \cdot \x}
		d\x. 
	\end{align}
\end{definition}

\begin{definition}[Three-quarter Fourier Transform]\label{def. 3/4FT}
	The three-quarter Fourier transform of a function $u$ is given by
	\begin{align}
		U_{3/4}(\boldsymbol{\alpha}) = \mathcal{F}_{3/4}[u](\boldsymbol{\alpha}) 
		= \iint_{\cup_{i=2}^4 Q_i} u(\x)e^{i \boldsymbol{\alpha} \cdot \x}
		d\x.
	\end{align}	
\end{definition}

Here, we have $\boldsymbol{\alpha} = (\alpha_1,\alpha_2) \in \mathbb{C}^2$ and we write $d\x$ for $dx_1dx_2$. More details as to where $\boldsymbol{\alpha}$ is permitted to go in $\mathbb{C}^2$ will be given in Section \ref{Cha:DomainsOfAnalyticty}. \RED{Recall the definitions of $\mathfrak{a}_1$ and $\mathfrak{a_2}$ given in \eqref{eq:IncidentRewritten}--\eqref{eq.a12Def}.}
Now, apply $\mathcal{F}_{1/4}$ to \eqref{eq:1.1} and $\mathcal{F}_{3/4}$ to \eqref{eq:1.2}. Using the boundary conditions \eqref{eq:1.3}--\eqref{eq:1.6} and setting
\begin{alignat}{3}
	& \Phi_{3/4}(\boldsymbol{\alpha})  =\mathcal{F}_{3/4}[\phi_{\ssc}], \   &&\Psi_{1/4}(\boldsymbol{\alpha})  =\mathcal{F}_{1/4}[\psi] , \label{eq.PhiDef}\\
	& P(\boldsymbol{\alpha}) = \frac{1}{(\alpha_1 -\RED{\mathfrak{a}_1})(\alpha_2 -\RED{\mathfrak{a}_2})}, \quad 	&& K(\boldsymbol{\alpha}) = \frac{k^2_2 - \alpha^2_1 - \alpha^2_2}{k^2_1 - \alpha^2_1 - \alpha^2_2}, \label{eq.KPDef}
\end{alignat} we find the following Wiener-Hopf equation (see Appendix \ref{Appendix: WH-Details} for the calculation):
				\begin{align*}
					- K(\boldsymbol{\alpha})\Psi_{1/4}(\boldsymbol{\alpha}) = \Phi_{3/4}(\boldsymbol{\alpha}) + P(\boldsymbol{\alpha}). \label{eq.1.22}
					\numberthis
				\end{align*}	

\begin{remark}[comparison with quarter-plane] \label{Remark:ComparisonWithQP}
	Note that \eqref{eq.1.22} is almost identical to the Wiener-Hopf equation for the quarter-plane given in \cite{AssierAbrahams1}. In fact, setting $\tilde{\Psi}_{1/4}= - \Psi_{1/4}$ we can rewrite \eqref{eq.1.22} as \begin{align}K(\boldsymbol{\alpha})\tilde{\Psi}_{1/4}(\boldsymbol{\alpha})= \Phi_{3/4} + P(\boldsymbol{\alpha}) \label{eq.quarter-planeWienerHopf}
	\end{align}
	which, formally, is the same Wiener-Hopf equation as for the quarter-plane (that is, \eqref{eq.quarter-planeWienerHopf} and the Wiener-Hopf equation in \cite{AssierAbrahams1} only differ by the definition of the kernel $K$\RED{, which for the quarter-plane is given by $K(\boldsymbol{\alpha}) = 1/\sqrt{k^2-\alpha^2_1-\alpha^2_2}$, where $k$ is the (only) wavenumber of the quarter-plane problem}).
\end{remark} 

\subsection{Domains of analyticity}\label{Cha:DomainsOfAnalyticty}

Whilst we have, formally, found a functional equation for the diffraction problem at hand, the domain in $\mathbb{C}^2$ where this equation is valid has not yet been discussed. This is the aim of the present section. 


\subsubsection{Set notations}\label{sec:SetNotations}

Before we begin discussing equation \eqref{eq.1.22}'s validity, let us introduce some notation which will be used extensively throughout the remainder of this article.  For any $\kappa_1 < \kappa_2 \in \mathbb{R}$ we define (see figure \ref{fig:UHPLHPS1}) 
\begin{gather*}
	\mathrm{UHP}(\kappa_{\RED{j}}) = \{z \in \mathbb{C}| \mathrm{Im}(z) > \kappa_{\RED{j}}\}, \ \RED{j=1,2;} \ \	\mathrm{LHP}(\kappa_{\RED{j}})  = \{z \in \mathbb{C}| \mathrm{Im}(z) < \kappa_{\RED{j}}\}, \ \RED{j=1,2;} \\
	\S(\kappa_1,\kappa_{2}) = \{z \in \mathbb{C}| \kappa_{1} < \mathrm{Im}(z) < \kappa_{2}\}.
\end{gather*}
Visually speaking, the upper half plane $\mathrm{UHP}(\kappa_{j})$ (resp. lower half plane $\mathrm{LHP}(\kappa_{j})$) consists of all points $z \in \mathbb{C}$ lying above (resp. below) the line given by $\{\text{Im}(z)=\kappa_j\}$ whereas the strip $\S(\kappa_2,\kappa_1)$ consists of all points between the lines $\{\text{Im}(z)=\kappa_2\}$ and $\{\text{Im}(z)=\kappa_1\}$. \RED{In particular, $\S(\kappa_2,\kappa_1) = \UHP(\kappa_2) \cap \LHP(\kappa_1)$.}
\begin{figure}[thbp]
	\centering
	\includegraphics[width=\textwidth]{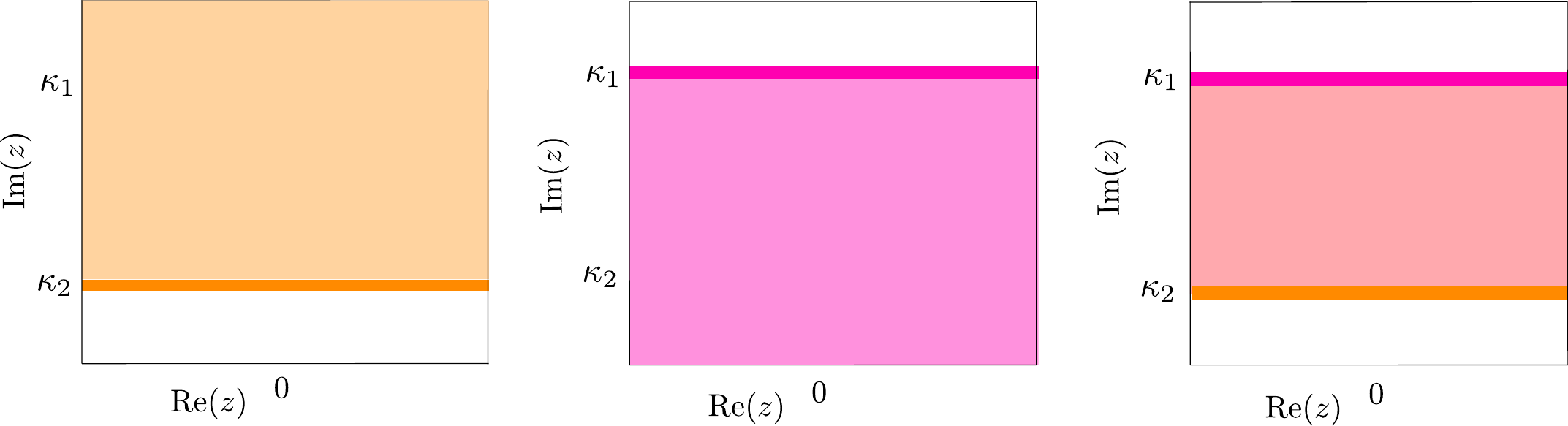}
	\vspace{-.25cm}
	\caption{Half planes $\UHP(\kappa_2)$ (left), $\LHP(\kappa_1)$ (middle), and strip $\S(\kappa_2,\kappa_1)$ (right).}
	\label{fig:UHPLHPS1}
\end{figure} 
Moreover, for any $\kappa_{1,2} \in \mathbb{R}$ (i.e.\! we now also allow $\kappa_1 \geq \kappa_2$) we define:
\begin{align*}
	&	\D_{++}(\kappa_1,\kappa_2) = \mathrm{UHP}(\kappa_1) \times \mathrm{UHP}(\kappa_2), \ \ \D_{-+}(\kappa_1,\kappa_2) = \mathrm{LHP}(\kappa_1) 
	\times \mathrm{UHP}(\kappa_2), \  \\
	&\D_{--}(\kappa_1,\kappa_2) = \mathrm{LHP}(\kappa_1) \times \mathrm{LHP}(\kappa_2), \  \	\D_{+-}(\kappa_1,\kappa_2) = \mathrm{UHP}(\kappa_1) \times \mathrm{LHP}(\kappa_2).
\end{align*}
\RED{In particular, $\D_{++}(\kappa_1,\kappa_2), \D_{-+}(\kappa_1,\kappa_2), \D_{--}(\kappa_1,\kappa_2),$ and $\D_{+-}(\kappa_1,\kappa_2)$ are (open) subsets of $\mathbb{C}^2$ i.e. if $(\alpha_1, \alpha_2) \in \D_{++}(\kappa_1,\kappa_2)$, say, then $\alpha_1 \in \mathrm{UHP}(\kappa_1)$ and $\alpha_2 \in \mathrm{UHP}(\kappa_2)$.}
\subsubsection{Function notations}\label{sec:FunctionNotations}

Using the sets defined above, we now introduce the following notation for functions:

\begin{definition}\label{def:++,--,+-,circFunctions}
	Let $U:D \to \mathbb{C}, \ D \subset \mathbb{C}^2$. We then call $U$ a $++,+-,--,$ or $-+$ function if, and only if, there are some $\kappa_{1,2}\in\mathbb{R}$ such that $U$ is analytic in $\D_{++}(\kappa_1,\kappa_2)$, $\D_{+-}(\kappa_1,\kappa_2)$, $\D_{--}(\kappa_1,\kappa_2)$, or $\D_{-+}(\kappa_1,\kappa_2)$. In these respective cases, we also write 
	\begin{align*}
		U=U_{++}, U_{+-}, U_{--},  U_{-+}.
	\end{align*}
	Moreover, if $U$ is analytic in $\S(\kappa'_1,\kappa'_2) \times \mathrm{UHP}(\kappa_2)$  for some $\kappa'_1 < \kappa'_2$ we write 
	\begin{align*}
		U=U_{\circ+}
	\end{align*}
	and say that $U$ is a $\circ+$ function. Analogously, the concepts of $\circ-, + \circ$ and $- \circ$ functions are defined and in these respective cases we write
	\begin{align*}
		U = U_{\circ-}, \ U = U_{+\circ}, \ \text{and} \ U= U_{- \circ}.
	\end{align*}
\end{definition}

\subsubsection{Domains of analyticity}\label{sec:Domains}

Recall that for half-range Fourier transforms, we have:

\begin{theorem}\label{Thm:FTAsympt}
	Let $f:\mathbb{R} \to \mathbb{C}$ satisfy  $|f(x)| < A e^{\RED{b_0} x}$ as $x \to \infty$ for some \RED{constants} 
	$b_0 \in \mathbb{R},\ \RED{A} \in \RED{[0,\infty)}$. Then the function $F_+(\alpha)$ defined by
	\begin{align*}
		F_+(\alpha)=\int_{0}^{\infty}f(x)e^{i \alpha x}dx \ 
	\end{align*}
	is analytic for all $\alpha \in \UHP(\RED{b_0})$.
	If, on the other hand, we have  $|f(x)| < A e^{\RED{b_0} x}$ as $x \to -\infty$ for some \RED{(maybe different) constants} 
	$\RED{b_0 \in \mathbb{R}, \ A} \in \RED{[0,\infty)}$  then the function $F_-(\alpha)$ defined by
	\begin{align*}
		F_-(\alpha)=\int_{-\infty}^{0}f(x)e^{i \alpha x}dx 
	\end{align*}	
	is analytic for all $\alpha \in \LHP(\RED{b_0})$. \RED{Note that the specific value of the constant $A$ is irrelevant for the analyticity behaviour of $F_+(\alpha)$ and $F_-(\alpha)$ respectively.}
\end{theorem}

These are well-known results and we refer to \cite{Noble} for a more detailed discussion. Now, using geometrical optics and writing
$u=u_{\go} + u_{\diff}$ for $u=\phi_{\ssc} + \phi_{\iin}$ or $u=\psi$, we know that in the far field the wave $u_{\go}$, consisting of the incident, reflected, and transmitted plane waves in their respective domains, will always dominate the diffracted field (since $u_{\diff}$ is \RED{an exponentially} decaying cylindrical wave). Recall that $\Im(\RED{\mathfrak{a}_1}) = \RED{\Im(k_1)} \cos(\vartheta_0)$ and  $\Im(\RED{\mathfrak{a}_2})= \RED{\Im(k_2)} \sin(\vartheta_0)$, so setting 
\begin{align}
	\delta = \min\{\RED{\Im(k_1)} |\cos(\vartheta_0)|, \RED{\Im(k_2)} |\sin(\vartheta_0)|\} \label{eq.deltaDEF}
\end{align} we have $\Im(\RED{\mathfrak{a}_{1,2}}) \leq - \delta < 0 $, 
and we therefore obtain
\begin{align}
	|u_{\go}| \leq A e^{-\delta|x_1| - \delta |x_2|}  \ \text{as} \ x_1,x_2 \to \pm \infty \ \text{in} \ \mathbb{R}^2 \label{eq.u_goFarField}
\end{align}
for some  constant $A \RED{\in [0,\infty)}$ \RED{(again, the exact value of $A$ does not matter)}. Moreover, it can be shown that $K(\boldsymbol{\alpha})$ is analytic in $\S(-\varepsilon,\varepsilon) \times \S(-\varepsilon,\varepsilon)$ for a suitable constant $\RED{\varepsilon \in (0, \delta]}$, see Lemma \ref{lemma:KpmcircDomain}.
For simplicity, let us henceforth omit a function's argument unless it is \emph{not} $\boldsymbol{\alpha}$, and let us set 
\begin{align*}
	&\D_{++} = \D_{++}(-\varepsilon,-\varepsilon), \ \D_{+-} = \D_{+-}(-\varepsilon,-\varepsilon), \\
	&\D_{--} = \D_{--}(-\varepsilon,-\varepsilon), \ \D_{-+} = \D_{-+}(-\varepsilon,-\varepsilon), \\
	&\S = \S(-\varepsilon, \varepsilon), \ \LHP = \LHP(\varepsilon), \ \UHP= \UHP(-\varepsilon).
\end{align*}
Then, applying Theorem \ref{Thm:FTAsympt} twice and using \eqref{eq.u_goFarField} we find:
\begin{alignat*}{3}
	\Psi_{++} & = \Psi_{1/4},   && \  \text{analytic on} \ \D_{++},\\
	\Phi_{-+} &  = \iint_{Q_2}\phi_{\ssc}(\x) e^{i \boldsymbol{\alpha} \x} d\x, && \  \text{analytic on}  \ \D_{-+}, \\
	\Phi_{--}&  = \iint_{Q_3}\phi_{\ssc}(\x) e^{i \boldsymbol{\alpha} \x} d\x,  && \ \text{analytic on} \ \D_{--}, \\ 	
	\Phi_{+-} &  = \iint_{Q_4}\phi_{\ssc}(\x) e^{i \boldsymbol{\alpha} \x} d\x, && \ \text{analytic on} \  \D_{+-}, \\
	P_{++} =  P & =  \frac{1}{(\alpha_1 - \RED{\mathfrak{a}_1})(\alpha_2 - \RED{\mathfrak{a}_2})}, \  && \ \text{analytic on} \ \D_{++}. 
\end{alignat*}
\RED{Note that $P=P_{++}$ is analytic in $\D_{++}$ since $\mathfrak{a}_1$ and $\mathfrak{a}_2$ are in $\LHP$.}
\RED{Now, since} \RED{by definition of $\Phi_{3/4}$ (see \eqref{eq.PhiDef}) we have} 
\begin{align}	
	\Phi_{3/4} = \Phi_{+-} + \Phi_{--} + \Phi_{-+},
\end{align}
we find:

\begin{corollary} The spectral function $\Psi_{++}$ is analytic in the region $\D_{++}$ whereas $\Phi_{3/4}$ is analytic in the region $\S \times \S$.
\end{corollary} 

Thus, since $K$ is analytic on $\S \times \S$ and $P_{++}$ is analytic on $\D_{++}$ we find that \eqref{eq.1.22} is valid in 
$\S\times \S$. To summarise:

\begin{corollary}\label{Corollary:RewrittenWHequation}
	The Wiener-Hopf equation \eqref{eq.1.22} can be rewritten as
				\begin{align*}
						-\Psi_{++}K =\Phi_{+-} + \Phi_{--} + \Phi_{-+} + P_{++}, \label{eq.WH2}
						\numberthis
				\end{align*}				
	and is valid on $\S \times \S$. 
\end{corollary}		

\RED{Equation \eqref{eq.WH2} represents a generalization of the classical  (one complex-variable) Wiener-Hopf equation  that appears, for instance, in the diffraction by a half-plane, see \cite{Noble}.}

\section{Factorisation of $K$}\label{Section3}
\subsection{Some useful functions}\label{sec:usefulFun}
As usual in complex analysis, functions defined on the real numbers might exhibit branch points when analytically continued onto the complex plane (c.f.\! \cite{Wegert}). This leads to the function being defined not on $\mathbb{C}$ but on some Riemann surface instead. However, for the purpose of the present work, we do not need this generality and the interested reader is referred to \cite{Wegert} for a more detailed discussion of the process of analytical continuation. When instead of working on the function's Riemann surface one wants to work on $\mathbb{C}$, branch cuts have to be introduced that is, we have to introduce lines of discontinuity of our function, \RED{but there is some arbitrariness involved in the specific choice of branch cuts}. In this section, we will specify some choice of branch cut for the complex square root function as well as the complex logarithm. These specific choices are the same as in \cite{AssierAbrahams1} (for the logarithm) and \cite{AssierShanin} (for the square root function). All of the following functions play a crucial role in the factorisation of $K$. Throughout the remainder of the article, we will extensively employ the method of phase portraits to visualise a complex function's properties in the spirit of \cite{Wegert}.

Let $\log(z)$ and $\sqrt{z}$ denote the standard complex logarithm and square root used by most mathematical software (Matlab, for instance). These functions correspond to the usual real logarithm and square root respectively when restricted onto $\mathbb{R}^+$ and have a branch cut along the negative real axis \RED{(i.e. $\arg(z) \in (-\pi,\pi]$)}.

We define $\mylog(z)$ as the logarithm with a branch cut diagonally down the third quadrant, see figure \ref{fig:LogAndMylog}. Practically, $\mylog(z)$ is obtained via the relation $\mylog(z) = \log(e^{-i \pi/4} z) + i \pi /4$.
\begin{figure}[thbp]
	\centering
	\includegraphics[width=.85\textwidth]{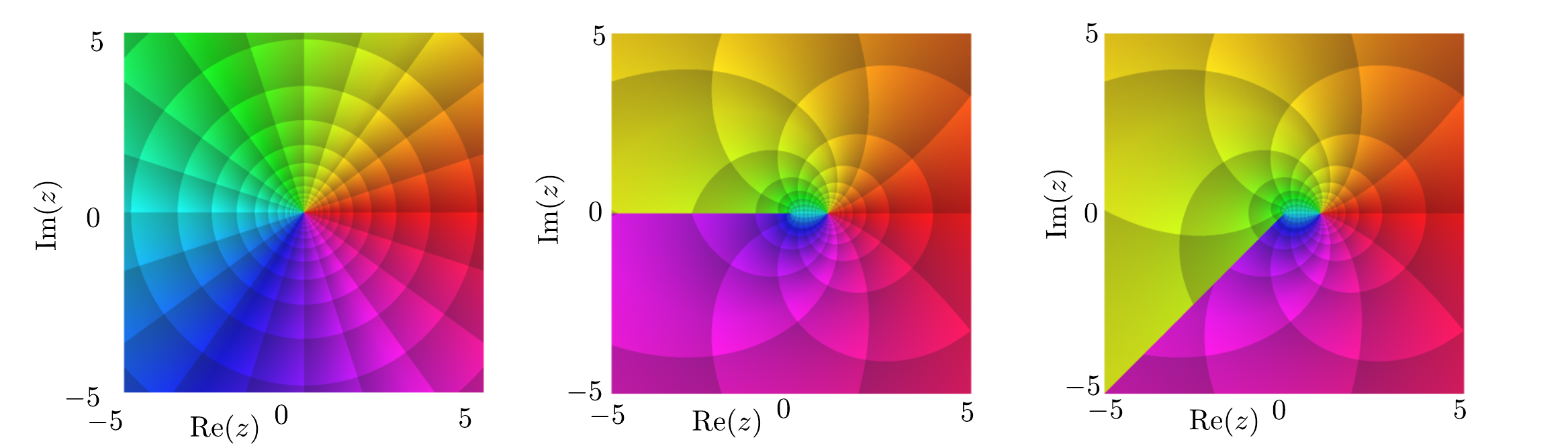}
	\vspace{-.5cm}
	\caption{Phase portrait of the functions $f(z)=z$ (left), $\log(z)$ (centre), and $\mylog(z)$ (right).}
	\label{fig:LogAndMylog}
\end{figure}

Next, we specify the choice of branch cut for the square root. Denote by $\mmysqrt{z}$ the square root function with branch cut along the positive real axis and branch subject to $\mmysqrt{-1} = i$, see figure \ref{fig:SqrtAndMysqrtKappa}. This choice of square root guarantees its imaginary part to be strictly positive everywhere except on the positive real axis (which is mapped onto the real line). Practically, $\mmysqrt{z}$ can be defined by $\mmysqrt{z} = i \sqrt{-z}$.
\begin{figure}[thbp]
	\centering
	\includegraphics[width=.85\textwidth]{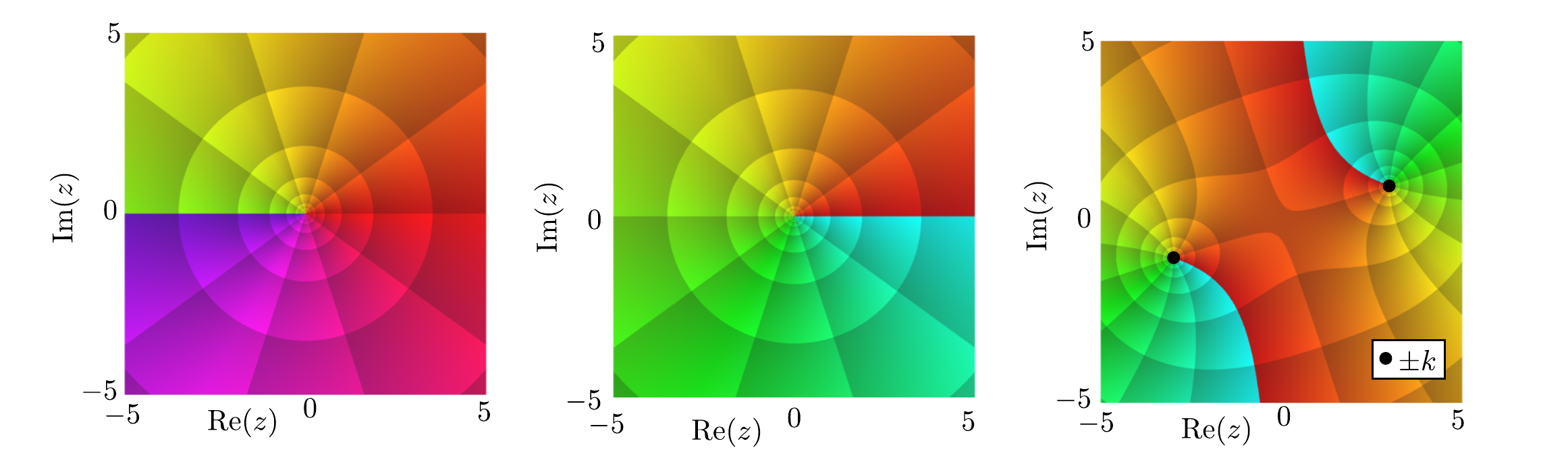}
	\vspace{-.4cm}
	\caption{Phase portraits of the functions $\sqrt{z}$ (left), $\mmysqrt{z}$ (centre), and $\kappa(k,z)$ (right) for $k=3+i$.}
	\label{fig:SqrtAndMysqrtKappa}
\end{figure}
Finally, for $k$ with $\Im(k) > 0$ and $\Re(k) > 0$ we define 
\begin{align}
	&	\kappa(k,z) = \mmysqrt{k^2 -z^2}  \label{eq.kappa}
\end{align}
which is visualised in figure \ref{fig:SqrtAndMysqrtKappa}. Due to the choice of square root, the sheet of $\kappa(k, \cdot)$'s Riemann surface is chosen such that $\kappa(k,0) = + k$.	
The function $\kappa(k,z)$ has two branch cuts, starting at $z=k$ and $z=-k$ respectively, see figure \ref{fig:SqrtAndMysqrtKappa}. Moreover, since $\mmysqrt{z}$ has strictly positive imaginary part everywhere except on its branch cut (where its imaginary part vanishes), $\kappa(k,z)$ also has strictly positive imaginary part everywhere except on its branch cuts (c.f.\! figure \ref{fig:SqrtAndMysqrtKappa}), which are mapped onto the real axis (see \cite{AssierShanin} for a more detailed discussion). 
\subsection{Factorisation in the $\alpha_1$ plane}\label{sec:alpha1K}

Recall the notation introduced in subsection \ref{sec:FunctionNotations}. 
Using $\kappa$, we can write 
\begin{align}
	K(\boldsymbol{\alpha}) = \frac{\left(\kappa(k_2,\alpha_2) + \alpha_1\right) \left(\kappa(k_2,\alpha_2) - \alpha_1\right)}{\left(\kappa(k_1,\alpha_2) + \alpha_1\right) \left(\kappa(k_1,\alpha_2) - \alpha_1\right)}. \label{eq.KFirstAlpha1Factor.}
\end{align}
Upon defining
\begin{align}
	K_{+\circ} = \frac{\kappa(k_2,\alpha_2) + \alpha_1}{\kappa(k_1,\alpha_2) + \alpha_1},  \ 
	K_{-\circ} = \frac{\kappa(k_2,\alpha_2) - \alpha_1}{\kappa(k_1,\alpha_2) - \alpha_1}, \label{eq.K+odef}
\end{align}	
and using \eqref{eq.KFirstAlpha1Factor.}, we have 
\begin{align}
	K = K_{+ \circ} K_{- \circ}. \label{eq.Kalpha1Fac}
\end{align}
The following lemma justifies the notation.
\begin{lemma}\label{lemma:KpmcircDomain}
	There exists an $\varepsilon>0$ such that $K_{+ \circ}$ and $K_{-\circ}$ are analytic in $\mathrm{UHP}(-\varepsilon) \times \S(-\varepsilon, \varepsilon)$ and $\mathrm{LHP}(\varepsilon) \times \S(-\varepsilon, \varepsilon)$ respectively. Note that this implies analyticity of $K$ in $\S(-\varepsilon,\varepsilon) \times \S(-\varepsilon,\varepsilon)$. \RED{Moreover, $K_{-\circ} \to 1$ and $K_{+\circ} \to 1$ as $|\alpha_{1,2}| \to \infty$ within these function's respective domains of analyticity.}
\end{lemma}

\RED{The domains of analyticity and the limiting behaviour of $K_{-\circ}$ and $K_{+\circ}$ will be crucial not only when factorising $K_{+ \circ}$ and $K_{- \circ}$ in the $\alpha_2$ plane in Section \ref{sec:alpha2plane} but also when applying Liouville's theorem in Section \ref{sec:WHSystem}.}

We only prove the lemma for $K_{-\circ}$ as the proof for $K_{+\circ}$ is analogous. See also figures \ref{fig:KMinusCirc1} and \ref{fig:KPlusCirc} for a visualisation, which will be explained in more detail below, after the proof.
\begin{proof}[Proof of Lemma \ref{lemma:KpmcircDomain}]
	Let us begin by examining the behaviour in the $\alpha_2$ plane, \RED{and let $\delta$ be as in \eqref{eq.deltaDEF}.} \RED{Since for $j=1,2$, the function $\alpha_1 \mapsto \kappa(k_j,\alpha_1)$} is analytic in $\S(-\delta,\delta)$ we only need to account for the polar singularities given by $\alpha_{2\text{sing}}$ such that $\kappa(k_{1},\alpha_{2\text{sing}}) = \alpha_1$. But due to the properties of $\kappa$, we know $\mathrm{Im}(\kappa(k_{1},\alpha_{2\text{sing}}))\geq0$ with equality only possible if $\Im(\alpha_{2 \text{sing}})\geq\Im(k_1) \geq \delta$. Therefore, if we restrict $\alpha_2 \in \S(-\delta/2,\delta/2)$, say, we obtain $\RED{\delta_1:=} \min_{\RED{\alpha_{2\text{sing}}}}\{\Im(\kappa(k_{1},\alpha_{2\text{sing}}))\}  > 0$.
	Choose $\varepsilon = \min\{\delta/2,\delta_1 \}$. \RED{The limiting behaviour at $\infty$ is directly obtained from the defining formula \eqref{eq.K+odef}.}
\end{proof}	

\begin{figure}[!h]
	\centering
	\includegraphics[width=.75\textwidth]{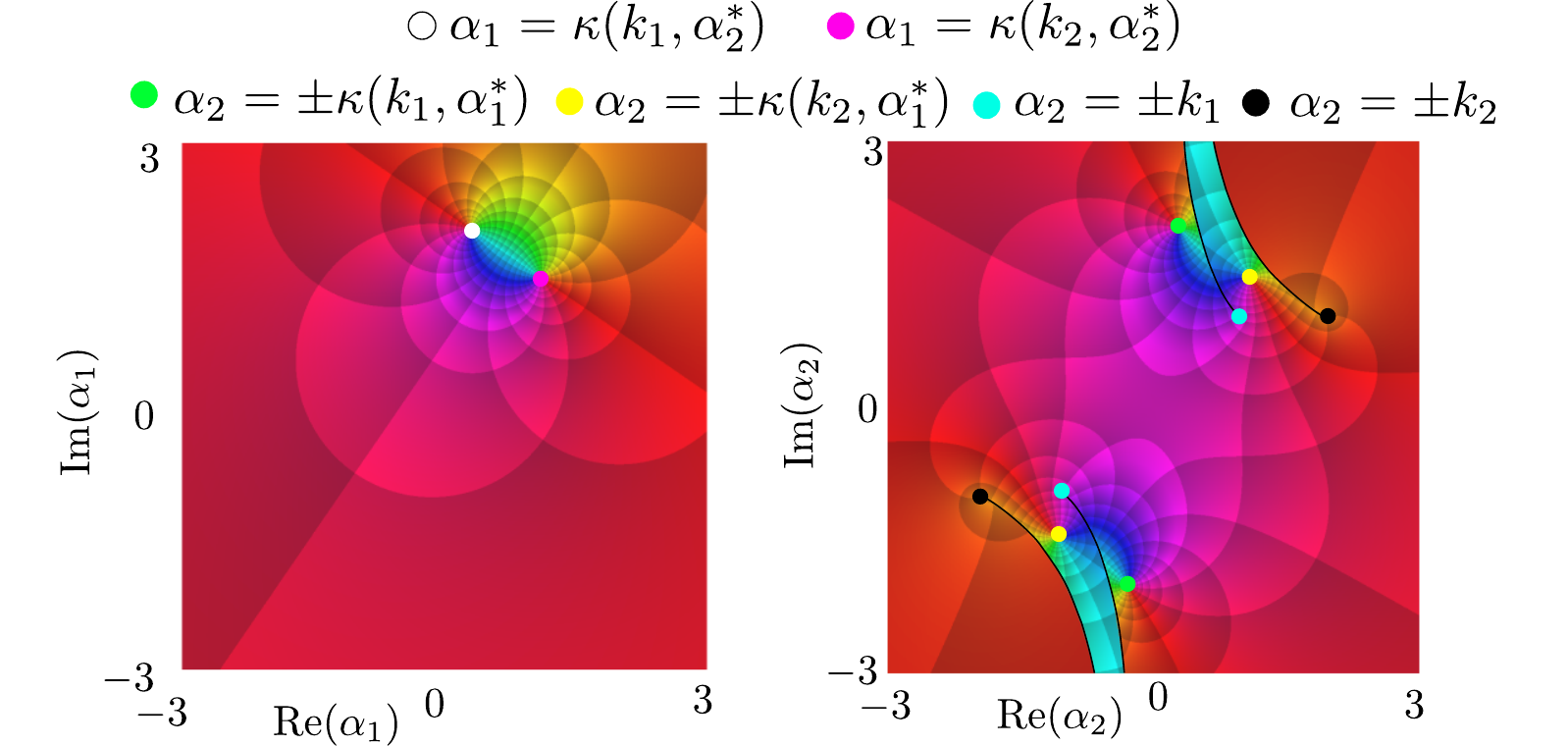}
	\vspace{-.25cm}
	\caption{Phase portraits of $K_{-\circ}$ for $k_1=1+i, \ k_2 = 2+i$. On the left, the phase portrait is taken in the $\alpha_1$ plane with fixed $\alpha^*_2 = 2 + \frac{1}{5}i$. On the right, it is taken in the $\alpha_2$ plane with fixed $\alpha^*_1=2+\frac{1}{5}i$.}
	\label{fig:KMinusCirc1}
\end{figure}
\begin{figure}[!h]
	\centering
	\includegraphics[width=.75\textwidth]{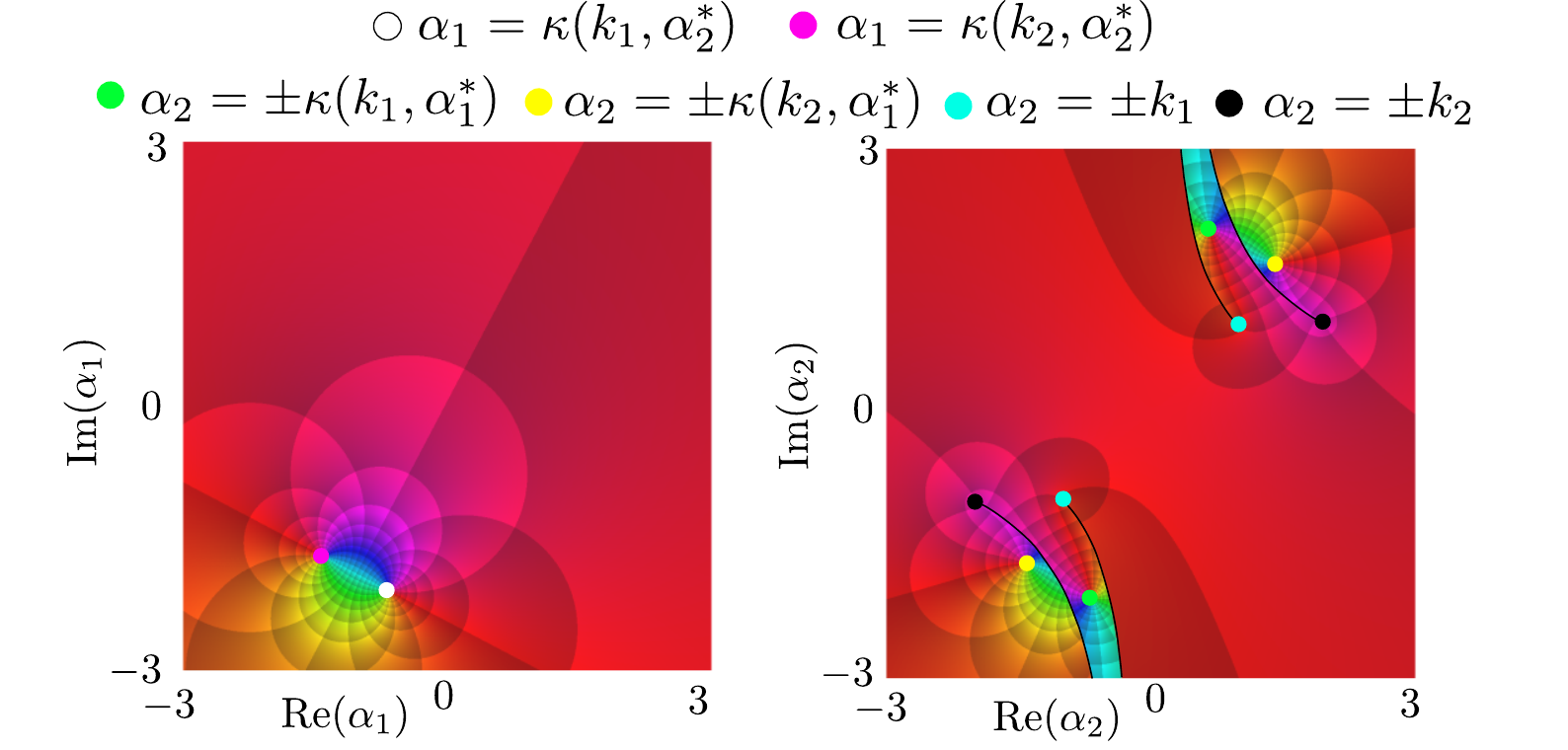}
	\vspace{-.25cm}
	\caption{Phase portraits of $K_{+\circ}$ for $k_1=1+i, \ k_2 = 2+i$. On the left, the phase portrait is taken in the $\alpha_1$ plane with fixed $\alpha^*_2 = 2 - \frac{1}{5}i$. On the right, it is taken in the $\alpha_2$ plane with fixed $\alpha^*_1=2-\frac{1}{5}i$.}
	\label{fig:KPlusCirc}
\end{figure}
Recall the notation $\S = \S(-\varepsilon,\varepsilon)$, $\UHP = \UHP(-\varepsilon)$, and $\LHP = \LHP(\varepsilon)$. Additionally, we define  
\begin{align}
	\D_{+\circ} = \UHP \times \S, \ \D_{- \circ} = \LHP \times \S, \ \D_{\circ +} = \S \times \UHP, \ \D_{\circ -} = \S \times \LHP \label{Definition:D+CircEtc}
\end{align}
so $K_{- \circ}$ is analytic on $\D_{- \circ}$ and $K_{+ \circ}$ is analytic on $\D_{+ \circ}$.

Figure \ref{fig:KMinusCirc1} visualises the properties of $K_{-\circ}$: We see that for fixed $\alpha^*_2 \in \S$ the function $K_{-\circ}(\alpha_1,\alpha^*_2)$ is analytic in the lower half plane, as the polar singularity corresponding to $\alpha_1 = \mmysqrt{k^2_1 - \alpha^{*2}_2} = \kappa(k_1,\alpha^*_2)$ lies in the upper half plane. For fixed $\alpha^*_1 \in \LHP$ on the other hand, we see that the function $K_{-\circ}(\alpha^*_1,\alpha_2)$ is analytic in some strip between its branch and polar singularities (located at $\alpha_2 = \pm \mmysqrt{k^2_1 - \alpha^{*2}_1} = \pm \kappa(k_1,\alpha^*_1)$ and $\alpha_2 = \pm k_{1,2}$ respectively) and that there are no polar singularities inside $\S$. 
An analogous visualisation of $K_{+ \circ}$ can be found in figure \ref{fig:KPlusCirc}. In figures \ref{fig:KMinusCirc1} and \ref{fig:KPlusCirc} respectively, the yellow points correspond to  polar singularities in the $\alpha_2$ plane, the white dot corresponds to the polar singularity in the $\alpha_1$ plane, whereas the cyan and black dots correspond to the function's branch points, and the green and magenta dots are simple zeros of the function.

\subsection{Factorisation in the $\alpha_2$ plane}\label{sec:alpha2K}
\subsubsection{Cauchy's formulae and bracket operators}

Throughout the remainder of this article, we will employ the following elementary yet essential theorems. These are classic results however, so we will omit the corresponding proofs. We refer to \cite{Noble} for a more detailed discussion. \RED{Moreover, all of this section's results also hold when the contour $\mathbb{R}$ which we use in the formulation of the following theorems and definition is replaced by a curved contour $\Gamma$, such as the contour $\Gamma$ mentioned in Section \ref{sec:VanishingImPart}, as long as the real part of $\Gamma$ starts at $-\infty$ and ends at $+\infty$, see \cite{AssierAbrahams1} and \cite{Noble}, but we do not need this generality for the context of the present article.} 

\RED{\begin{definition}
	We define the contours \RED{$\mathbb{R} - i \varepsilon$ and  $\mathbb{R}+i\varepsilon $ as 
		\begin{align*}
		\mathbb{R} - i \varepsilon = \{z \in \mathbb{C}| \ z=x - i \varepsilon, \ x \in \mathbb{R} \} \text{ and } \mathbb{R}+i\varepsilon= \{z \in \mathbb{C}| \ z=x + i \varepsilon, \ x \in \mathbb{R} \}
		\end{align*}} oriented from left to right for \RED{$\varepsilon$ as in Lemma \ref{lemma:KpmcircDomain}.}
\end{definition}}

\begin{theorem}[Cauchy's Formula; Sum-split]\label{thm. Cauchy Sum-split} \ 
	Let $\Phi$ be a function analytic on  \RED{$\S$, where $\S=\S(-\varepsilon,\varepsilon)$ is as defined in Section \ref{sec:Domains}}. 
	Then, provided  $\Phi(\alpha) \to 0$ as $|\alpha| \rightarrow \infty$ within
	\RED{$\S$} we have $\Phi(\alpha)=\Phi_{+}(\alpha)+\Phi_{-}(\alpha)$ on \RED{$\S$} with $\Phi_{+}$ analytic on the upper half plane \RED{$\mathrm{UHP}$} and $\Phi_{-}$ analytic on the lower half plane \RED{$\mathrm{LHP}$}. 
 \RED{Specifically, for $\alpha \in \S$ we have}
	\begin{align*}
		\Phi_{+}(\alpha)=\frac{1}{2 i \pi} \int_{\RED{\mathbb{R} -i\varepsilon}} \frac{\Phi(z)}{z-\alpha} \mathrm{d} z \ \text{ and } \	 \Phi_{-}(\alpha)=\frac{-1}{2 i \pi} \int_{\RED{\mathbb{R} +i\varepsilon}} \frac{\Phi(z)}{z-\alpha} \mathrm{d} z,
	\end{align*}
	and these formulae can be used to analytically continue $\Phi_{+}$ (resp. $\Phi_{-}$) onto \RED{$\mathrm{UHP}$} (resp. \RED{$\mathrm{LHP}$}).	
\end{theorem}

Following \cite{AssierAbrahams1}, using Cauchy's sum split we can define the following bracket operators: 
\begin{definition}[Bracket Operators]\label{def:BracketOperators}
	For any function $F: \S \times \S \to \mathbb{C}$ satisfying the conditions of Theorem \ref{thm. Cauchy Sum-split} in the $\alpha_1$ plane, say,  we define $[F]_{+ \circ}$ and $[F]_{- \circ}$ \RED{(analytic in $\D_{+\circ}$ and $\D_{-\circ}$ respectively)} as 
	\begin{align*}
		[F]_{+ \circ} = \frac{1}{2\pi i} \int_{\RED{\mathbb{R} - i \varepsilon}} \frac{F(z,\alpha_2)}{z-\alpha_1} dz \ \text{ and } \ [F]_{- \circ} = \frac{-1}{2\pi i} \int_{\RED{\mathbb{R} + i \varepsilon}} \frac{F(z,\alpha_2)}{z-\alpha_1} dz
	\end{align*}
	\RED{and in particular, we have $F= [F]_{+ \circ} + [F]_{- \circ}$ on $\S \times \S$}.
	Similarly, we define $[F]_{\circ +}$ and $[F]_{\circ -}$ if $F$ satisfies the conditions of Theorem \ref{thm. Cauchy Sum-split} in the $\alpha_2$ plane\RED{, and we have $F= [F]_{+ \circ} + [F]_{- \circ}$ on $\S \times \S$.} Note that $F$ can be defined on a domain larger than $\S \times \S$. 
\end{definition}

\begin{theorem}[Cauchy's Formula; Factorisation]\label{thm. Cauchy-Fact.} \ Let  $\Psi$ be a function analytic on $\S$ \RED{such that $\Psi$ has no zeros in $\mathcal{S}$ and $\Psi \to 1$ as $|\alpha| \to \infty$ within $\mathcal{S}$. Upon choosing the principal branch of the $\log$, this implies that $\log{\Psi} \to 0$ as $|\alpha| \to \infty$ within $\mathcal{S}$.} Then we have $\Psi(\alpha)=\Psi_{+}(\alpha) \Psi_{-}(\alpha)$ on $\S$ with $\Psi_{+}$ analytic on \RED{$\UHP$} and $\Psi_{-}$ analytic on \RED{$\LHP$}. \RED{Specifically,} for \RED{$\alpha \in \S$} we have
{\fontsize{10}{0}\selectfont	\begin{align*}
	\Psi_{+}(\alpha)=\exp \left(\frac{1}{2 i \pi} \int_{\RED{\mathbb{R} - i \varepsilon}} \frac{\RED{\log} (\Psi(z))}{z-\alpha} \mathrm{d} z\right) \text {and } \Psi_{-}(\alpha)=\exp \left(\frac{-1}{2 i \pi} \int_{\RED{\mathbb{R} + i \varepsilon}} \frac{\RED{\log} (\Psi(z))}{z-\alpha} \mathrm{d} z\right)
	\end{align*}}and these formulae can be used to analytically continue $\Psi_{+}$ onto \RED{$\UHP$} and $\Psi_{-}$  onto \RED{$\LHP$}.
\end{theorem}


\subsubsection{Factorisation of $K_{+ \circ}$ and $K_{- \circ}$ in the $\alpha_2$ plane}\label{sec:alpha2plane}

We wish to factorise $K_{+ \circ}$ and $K_{- \circ}$ in the $\alpha_2$ plane. Thus we need to verify the conditions of Theorem \ref{thm. Cauchy-Fact.}.

First, note that for fixed $\alpha^*_1$ we have 
\begin{align}
	K_{\pm \circ}(\alpha^*_1,\alpha_2) \to 1, \ \text{as} \ |\alpha_2| \to \infty \ \text{in} \ \S.
\end{align}
Thus, we just have to verify that that $K_{\pm \circ}$ does not cross $\mylog$'s branch cut, i.e. that $\mylog(K_{\pm\circ}(\alpha^*_1,\alpha_2))$
is analytic for all $\alpha_2 \in \S$. It is possible to prove this rigorously, but this is rather technical. Therefore,\RED{ in the spirit of \cite{AssierAbrahams1},} we instead provide a \emph{visual proof} of analyticity\RED{, which illustrates the validity of the statement}. 
Indeed, from figure \ref{fig:MylogKPlusCircMylogKMinusCirc} (top) we see that $\mylog(K_{-\circ})$ has no singularities for 
$\boldsymbol{\alpha} \in \LHP \times \S$ and is therefore, in particular, well-defined on $\S(-\varepsilon,\varepsilon)$ in the $\alpha_2$ plane (where $\varepsilon$ is as in Lemma \ref{lemma:KpmcircDomain}). 
Similarly, we see that $K_{+\circ}$ satisfies the conditions of Theorem \ref{thm. Cauchy-Fact.} in figure \ref{fig:MylogKPlusCircMylogKMinusCirc} (bottom). 
\begin{figure}[h!]
	\centering
	\includegraphics[width=.75\textwidth]{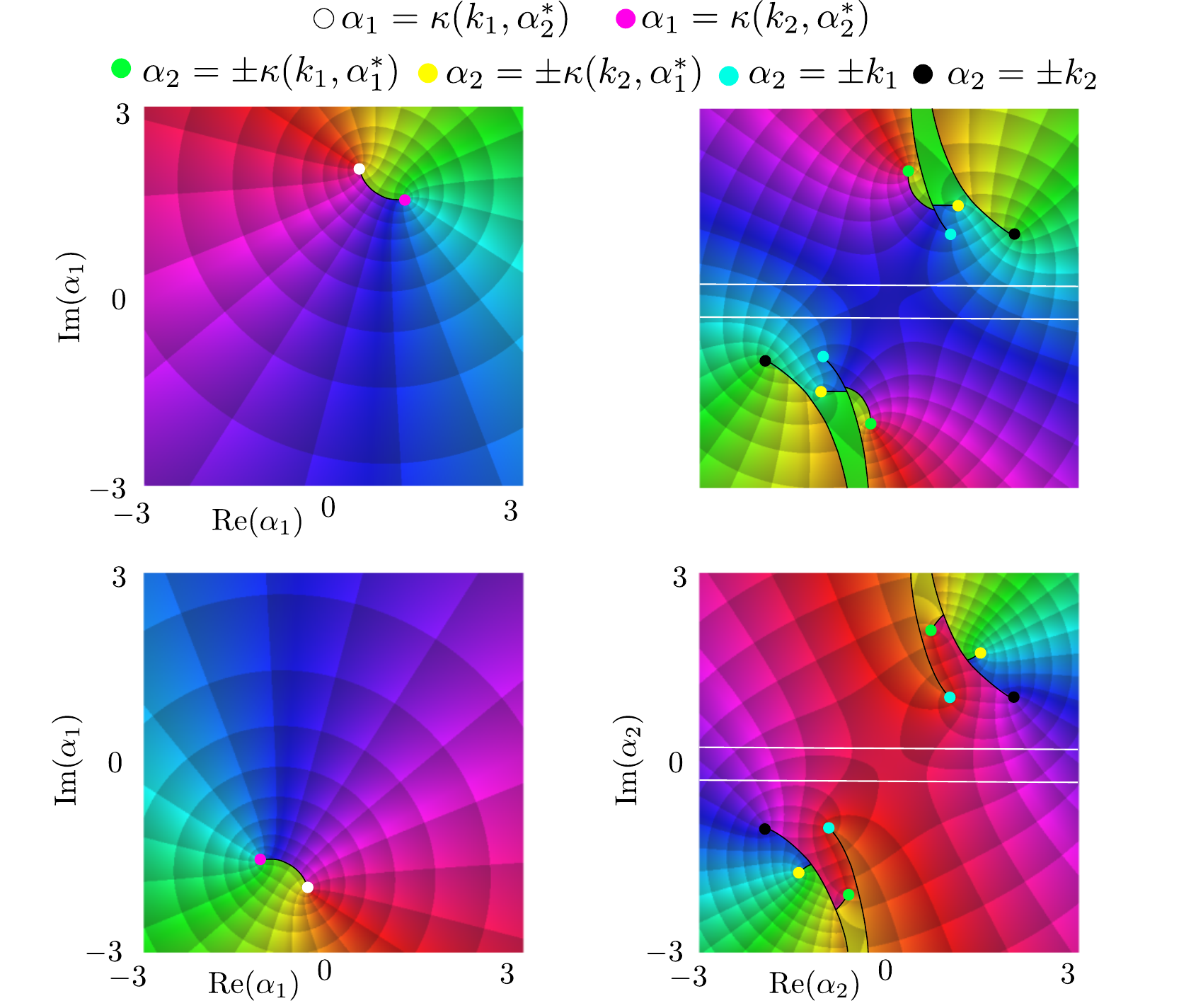}
	\vspace{-.25cm}
	\caption{Phase portrait of $\mylog(K_{-\circ})$ (top) with parameters as in figure \ref{fig:KMinusCirc1}, and phase portrait of $\mylog(K_{+\circ})$ (bottom) with parameters as in figure \ref{fig:KPlusCirc}. The contours $\mathbb{R} \pm i \varepsilon$ in the $\alpha_2$ plane are shown in white}
\label{fig:MylogKPlusCircMylogKMinusCirc}
\end{figure}
Therefore, we may apply Theorem \ref{thm. Cauchy-Fact.} to $K_{+\circ}$ and $K_{-\circ}$ and obtain
\vspace{-.1cm}
\begin{align}
K_{-\circ} = K_{--}K_{-+}, \ K_{+\circ} = K_{++}K_{+-}, \label{KpmCircFactorisation}
\end{align}
\vspace{-.15cm}
where 
\begin{align}
&	K_{--}(\boldsymbol{\alpha})  =  \exp\left(\frac{-1}{2\pi i} \int_{\mathbb{R}+i\varepsilon}  \frac{\mylog(K_{-\circ}(\alpha_1,z))}{z-\alpha_2} dz\right),		 
\end{align}
\begin{align}
&	K_{-+} (\boldsymbol{\alpha})  = \exp\left(\frac{1}{2\pi i} \int_{\mathbb{R}-i\varepsilon}  \frac{\mylog(K_{-\circ}(\alpha_1,z))}{z-\alpha_2}dz\right),		
\end{align}
\begin{align}
&K_{+-}(\boldsymbol{\alpha}) = \exp\left(\frac{-1}{2\pi i} \int_{\mathbb{R}+i\varepsilon}  \frac{\mylog(K_{+\circ}(\alpha_1,z))}{z-\alpha_2}dz\right),
\end{align}
\begin{align}
&K_{++}(\boldsymbol{\alpha}) =  \exp\left(\frac{1}{2\pi i} \int_{\mathbb{R}-i\varepsilon}  \frac{\mylog(K_{+\circ}(\alpha_1,z))}{z-\alpha_2}dz\right).
\end{align}	
By construction, we hence have
\begin{align}
K = K_{++}K_{+-}K_{--}K_{-+} \ \text{on} \ \S \times \S \label{eq.KFourWayFac}
\end{align}		
and we can verify the multiplicative structure \eqref{KpmCircFactorisation} in figures \ref{fig:KMinusCircMult} and \ref{fig:KPlusCircMult} respectively. Here, we chose to visualise the functions in the $\alpha_2$ plane but, as previously, it is of course also possible to visualise them in the $\alpha_1$ plane. 

\RED{\begin{remark}
	We may equally well choose to first factorise the kernel $K$ in the $\alpha_2$ plane and thereafter in the $\alpha_1$ plane. The procedure for doing this is exactly the same as the procedure discussed in Sections \ref{sec:alpha1K}--\ref{sec:alpha2K}, and will lead to a factorisation $K=\tilde{K}_{++}\tilde{K}_{+-}\tilde{K}_{--}\tilde{K}_{-+}$. By an application of Liouville's theorem, it can be shown $\tilde{K}_{++}=K_{++}$ etc., and therefore the resulting factorisation of $K$ given in \eqref{eq.KFourWayFac} does not depend on whether we first factorise in the $\alpha_1$ or $\alpha_2$ plane.
\end{remark}}

\begin{figure}[h!]
\centering
\includegraphics[width=\textwidth]{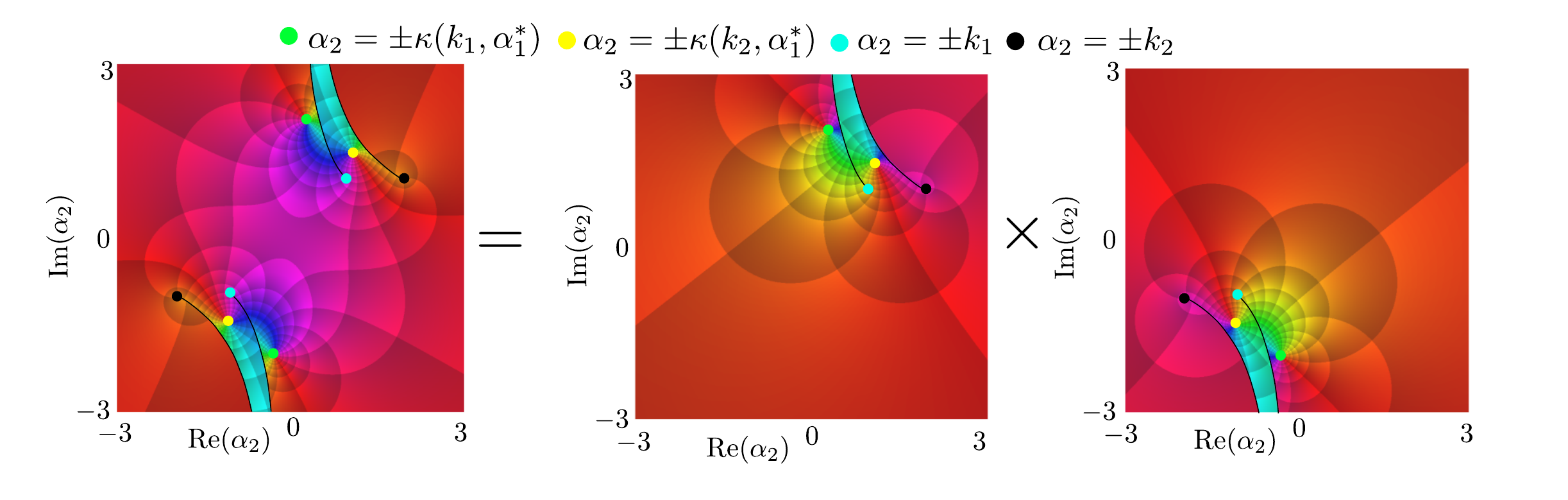}
\vspace{-.3cm}
\caption{Visualisation of $K_{-\circ} = K_{--}K_{-+}$ in the $\alpha_2$ plane with parameters as in figure \ref{fig:KMinusCirc1}. $K_{--}$ is shown in the middle and $K_{-+}$ is shown on the right.}
\label{fig:KMinusCircMult}
\end{figure}

\begin{figure}[h]
\centering
\includegraphics[width=\textwidth]{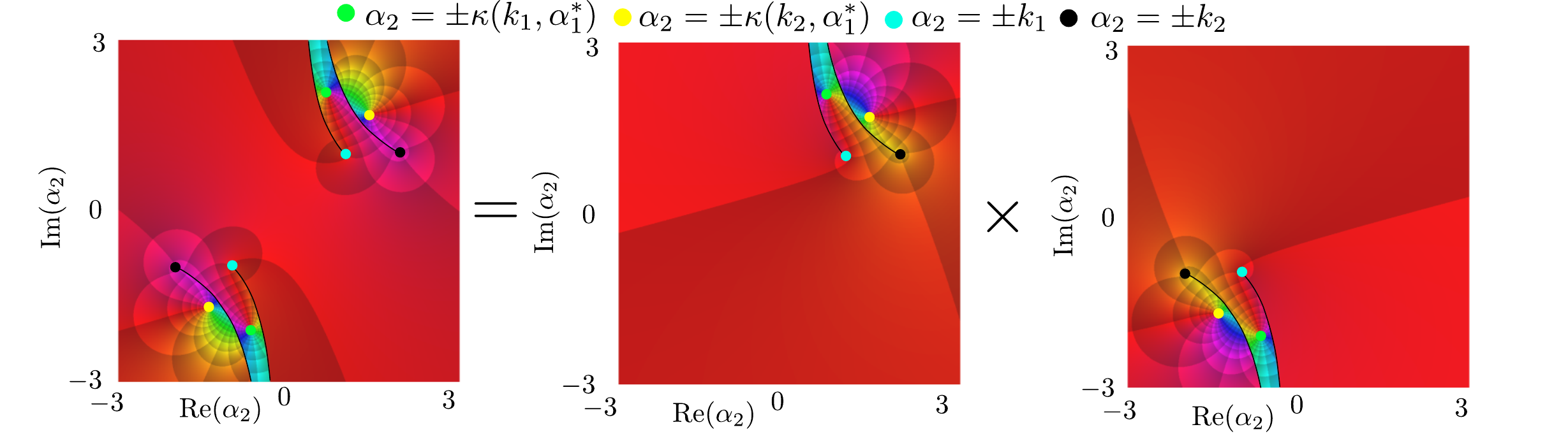}
\vspace{-.3cm}
\caption{Visualisation of $K_{+\circ} = K_{+-}K_{++}$ in the $\alpha_2$ plane with parameters as in figure \ref{fig:KPlusCirc}. $K_{+-}$ is shown in the middle and $K_{++}$ is shown on the right.}
\label{fig:KPlusCircMult}
\end{figure}

\section{The Wiener-Hopf system in $\mathbb{C}^2$}\label{sec:WHSystem}\label{Section4}

Recall the notions of $++$ functions, $+-$ functions, etc.\! (c.f.\! Definition \ref{def:++,--,+-,circFunctions}) and recall that by Corollary \ref{Corollary:RewrittenWHequation}, using these notations, the Wiener-Hopf equation \eqref{eq.1.22} can be rewritten as 
\begin{align}
-\Psi_{++}K = \Phi_{+-} + \Phi_{--} + \Phi_{-+} + P_{++} \label{eq.WH3}.
\end{align}
In the following two subsections, we will show how \eqref{eq.WH3} can be reduced to two coupled equations involving the unknowns $\Psi_{++}$ and $\Phi_{+-}$. This heavily relies on the kernel's factorisation and the bracket operators (c.f.\!  Definition \ref{def:BracketOperators}). Recall that we omit a function's argument \emph{unless it is not $\boldsymbol{\alpha}$}.

\subsection{Split in the $\alpha_1$ plane}\label{sec:WHalpha1}

We begin by writing \eqref{eq.WH3} as
\begin{align}
-K_{+\circ} \Psi_{++} = \frac{\Phi_{-\circ}}{K_{- \circ}} + 
\frac{\Phi_{+-}}{K_{- \circ}} + \frac{P_{++}}{K_{-\circ}}.
\end{align}
where we have set $\Phi_{-+} + \Phi_{--} = \Phi_{- \circ}$\RED{ and used the representation $K=K_{+\circ}K_{-\circ}$ given in \eqref{eq.Kalpha1Fac}.} For now, we just assume that Cauchy's \RED{formulas} \ref{thm. Cauchy Sum-split} and \ref{thm. Cauchy-Fact.} may be applied as we do below. This is possible due to the edge conditions \eqref{eq.2.56} and \eqref{eq.2.57} and the duality of near field behaviour in physical space and far field behaviour in Fourier space. We postpone the technical details to Appendix \ref{sec.LiouvilleApplication}. Now, applying the Cauchy sum-split to $\Phi_{+-}/K_{-\circ}$ in the $\alpha_1$ plane we obtain 
\begin{align}
-K_{+\circ} \Psi_{++} -\left[\frac{\Phi_{+-}}{K_{- \circ}}\right]_{+\circ}\!\! -\frac{P_{++}}{K_{-\circ}} = \frac{\Phi_{-\circ}}{K_{- \circ}} + \left[\frac{\Phi_{+-}}{K_{- \circ}}\right]_{-\circ}\!\!. \label{eq.3.10}
\end{align}
\RED{Recall the notations $\mathfrak{a}_1=k_1\cos(\vartheta_0) \in \LHP$ and $\mathfrak{a}_2=k_1\sin(\vartheta_0) \in \LHP$, introduced in \eqref{eq:IncidentRewritten}, and recall $P_{++} = \frac{1}{(\alpha_1 - \mathfrak{a}_1)(\alpha_2 - \mathfrak{a}_2)}$ (see \eqref{eq.KPDef}).} Now, by pole removal \RED{in the $\alpha_1$ plane}
\begin{align*}
\frac{P_{++}}{K_{-\circ}} = \underbrace{\frac{P_{++}}{K_{-\circ}(\RED{\mathfrak{a}_1},\alpha_2)}}_{\RED{\text{analytic in } \D_{+\circ}}} + \underbrace{P_{++}\left(\frac{1}{K_{-\circ}} - \frac{1}{K_{-\circ}(\RED{\mathfrak{a}_1},\alpha_2)} \right)}_{\RED{\text{analytic in } \D_{- \circ}}}.  
\end{align*}
\RED{Analyticity of the first term in $\D_{+\circ}$ is simple: the denominator does not depend on $\alpha_1$ and the numerator is analytic. For the second term the polar singularity is effectively removed since 
\[\left(\frac{1}{K_{-\circ}}-\frac{1}{K_{-\circ}(\mathfrak{a}_1,\alpha_2)}\right) \sim \frac{\kappa(k_1,\alpha_2) - \kappa(k_2,\alpha_2)}{(\kappa(k_2,\alpha_2) - \mathfrak{a}_1)^2}(\alpha_1-\mathfrak{a}_1) \text{ as } \alpha_1 \to \mathfrak{a}_1,\] which proves analyticity in $\D_{-\circ}$.}
Therefore \eqref{eq.3.10} is equivalent to 
\small
\begin{align*}
-K_{+\circ} \Psi_{++} -\left[\frac{\Phi_{+-}}{K_{- \circ}}\right]_{+\circ}\hspace{-.3cm}-	\frac{P_{++}}{K_{-\circ}(\RED{\mathfrak{a}_1},\alpha_2)} = 
\frac{\Phi_{-\circ}}{K_{- \circ}} + \left[\frac{\Phi_{+-}}{K_{- \circ}}\right]_{-\circ}\hspace{-.45cm}+ P_{++}\left(\frac{1}{K_{-\circ}} - 	\frac{1}{K_{-\circ}(\RED{\mathfrak{a}_1},\alpha_2)} \right), \numberthis \label{eq.3.10*}
\end{align*}
\normalsize
and the LHS of  \eqref{eq.3.10*} is analytic in $\D_{+ \circ}$ whereas the RHS is analytic in $\D_{- \circ}$. Thus, we can use this equality to obtain a function $E_1$ analytic on $\mathbb{C} \times \S$ by 
\begin{align*}
E_1(\alpha_1,\alpha_2) = \begin{cases}
	-K_{+\circ} \Psi_{++} -\left[\frac{\Phi_{+-}}{K_{- \circ}}\right]_{+\circ}\!\!  -\frac{P_{++}}{K_{-\circ}(\RED{\mathfrak{a}_1},\alpha_2)},&\text{if} \ \boldsymbol{\alpha} \in \D_{+\circ}, \\
	\frac{\Phi_{-\circ}}{K_{- \circ}} + \left[\frac{\Phi_{+-}}{K_{- \circ}}\right]_{-\circ}\!\!+
	P_{++}\left(\frac{1}{K_{-\circ}} - \frac{1}{K_{-\circ}(\RED{\mathfrak{a}_1},\alpha_2)} \right),\!\!& \text{if} 
	\   \boldsymbol{\alpha} \in \D_{-\circ}.
\end{cases}
\end{align*}
It can be shown that we can apply Liouville's theorem in the $\alpha_1$ plane (see Lemma \ref{Lemma:LiouvilleAlpha1}) and we find $E_1 \equiv 0$ for $\alpha_2 \in \S$. Therefore
			\begin{align}
				K_{+\circ} \Psi_{++} +\left[\frac{\Phi_{+-}}{K_{- \circ}}\right]_{+\circ}\!\!+ \frac{P_{++}}{K_{-\circ}(\RED{\mathfrak{a}_1},\alpha_2)}& = 0, \ \boldsymbol{\alpha} \in \D_{+\circ}, \label{eq.3.11}\\
				\frac{\Phi_{-\circ}}{K_{- \circ}} + \left[\frac{\Phi_{+-}}{K_{- \circ}}\right]_{-\circ}\!\!+ P_{++}\left(\frac{1}{K_{-\circ}} - \frac{1}{K_{-\circ}(\RED{\mathfrak{a}_1},\alpha_2)} \right)& = 0, \ \boldsymbol{\alpha} \in \D_{-\circ}. \label{eq.3.12}
			\end{align}  	

\subsection{Split in the $\alpha_2$ plane}
Multiplying \eqref{eq.3.11} by $K_{-+}(\RED{\mathfrak{a}_1},\alpha_2)/K_{+-}$ \RED{and using \eqref{KpmCircFactorisation}} we obtain 
\begin{align*}
-\Psi_{++}K_{++}K_{-+}(\RED{\mathfrak{a}_1},\alpha_2) = &  \frac{P_{++}}{K_{--}(\RED{\mathfrak{a}_1},\alpha_2)K_{+-}} 
+\frac{K_{-+}(\RED{\mathfrak{a}_1},\alpha_2)}{K_{+-}}\left[\frac{\Phi_{+-}}{K_{-\circ}}\right]_{+\circ}
\label{eq.3.27} \numberthis
\end{align*}
which is valid in $\D_{+ \circ}$. 
Applying the Cauchy sum-split in the $\alpha_2$ plane to  $\frac{K_{-+}(\RED{\mathfrak{a}_1},\alpha_2)}{K_{+-}}\left[\frac{\Phi_{+-}}{K_{-\circ}}\right]_{+\circ}$ 
we obtain 
\small
\begin{align} \frac{K_{-+}(\RED{\mathfrak{a}_1},\alpha_2)}{K_{+-}}\left[\frac{\Phi_{+-}}{K_{-\circ}}\right]_{+\circ}\!\!
= \left[\frac{K_{-+}(\RED{\mathfrak{a}_1},\alpha_2)}{K_{+-}}\left[\frac{\Phi_{+-}}{K_{-\circ}}\right]_{+\circ}\right]_{\circ-}\!\!+ 
\left[\frac{K_{-+}(\RED{\mathfrak{a}_1},\alpha_2)}{K_{+-}}\left[\frac{\Phi_{+-}}{K_{-\circ}}\right]_{+\circ}\right]_{\circ+}\!\!.
\end{align}	
\normalsize	
Similarly 
\begin{align}
\frac{P_{++}}{K_{--}(\RED{\mathfrak{a}_1},\alpha_2)K_{+-}} = \left[\frac{P_{++}}{K_{--}(\RED{\mathfrak{a}_1},\alpha_2)K_{+-}}\right]_{\circ -} + \left[\frac{P_{++}}{K_{--}(\RED{\mathfrak{a}_1},\alpha_2)K_{+-}}\right]_{\circ +}
\end{align}
and by pole removal \RED{in the $\alpha_2$ plane}: 
\begin{align*}
&	\left[\frac{P_{++}}{K_{--}(\RED{\mathfrak{a}_1},\alpha_2)K_{+-}}\right]_{\circ-}  = P_{++} \left(
\frac{1}{K_{--}(\RED{\mathfrak{a}_1},\alpha_2) K_{+-}} - \frac{1}{K_{--}(\RED{\mathfrak{a}_1},\RED{\mathfrak{a}_2})K_{+-}(\alpha_1,\RED{\mathfrak{a}_2})}\right), \\ 
&	\left[\frac{P_{++}}{K_{--}(\RED{\mathfrak{a}_1},\alpha_2)K_{+-}}\right]_{\circ+} = \frac{P_{++}}{K_{--}(\RED{\mathfrak{a}_1},\RED{\mathfrak{a}_2}) K_{+-}(\alpha_1,\RED{\mathfrak{a}_2})}. 
\end{align*}
\RED{Similarly to the pole removal performed in Section \ref{sec:WHalpha1}, the analyticity of $P_{++}/K_{--}(\RED{\mathfrak{a}_1},\RED{\mathfrak{a}_2})$ in $\D_{\circ -}$ is verified, and the analyticity of \[
	P_{++} \left(
	\frac{1}{K_{--}(\RED{\mathfrak{a}_1},\alpha_2) K_{+-}} - \frac{1}{K_{--}(\RED{\mathfrak{a}_1},\RED{\mathfrak{a}_2})K_{+-}(\alpha_1,\RED{\mathfrak{a}_2})}\right)
	\] in $\D_{\circ +}$ can be proved by writing $1/K_{--}(\RED{\mathfrak{a}_1},\alpha_2) K_{+-}$ as its Taylor series (in the $\alpha_2$ plane) at $\mathfrak{a}_2$.} Therefore, we can use \eqref{eq.3.27} to obtain a function $E_2$ analytic on $\UHP \times \mathbb{C}$ by
 {\fontsize{7.85pt}{0pt}\selectfont
		\begin{align}
			\nonumber
			E_{2}=  \begin{cases}
				-\Psi_{++}K_{++}K_{-+}(\RED{\mathfrak{a}_1},\alpha_2)  - \frac{P_{++}}{K_{--}(\RED{\mathfrak{a}_1},\RED{\mathfrak{a}_2})K_{+-}(\alpha_1,\RED{\mathfrak{a}_2})}
				-  \left[\frac{K_{-+}(\RED{\mathfrak{a}_1},\alpha_2)}{K_{+-}} \left[\frac{\Phi_{+-}}{K_{-\circ}}\right]_{+\circ}\right]_{\circ +}\!\!, \ \boldsymbol{\alpha} \in \D_{++} \\
				P_{++}\left(
				\frac{1}{K_{--}(\RED{\mathfrak{a}_1},\alpha_2) K_{+-}} \right. - \left. \frac{1}{K_{--}(\RED{\mathfrak{a}_1},\RED{\mathfrak{a}_2})K_{+-}(\alpha_1,\RED{\mathfrak{a}_2})}\right)
				+  \left[\frac{K_{-+}(\RED{\mathfrak{a}_1},\alpha_2)}{K_{+-}} \left[\frac{\Phi_{+-}}{K_{-\circ}}\right]_{+\circ}\right]_{\circ -}\!\!\!\!, \  \boldsymbol{\alpha} \in \D_{+-}. 
			\end{cases}
		\end{align}
	}
\RED{Similar to Section \ref{sec:WHalpha1}}, it can be shown that we can apply Liouville's theorem in the $\alpha_2$ plane to $E_2$ and obtain $E_2 \equiv 0$ (see Lemma \ref{Lemma:LiouvilleAlpha2}). Therefore we find the main result of the present work:

\begin{theorem}\label{MainThm} The unknowns $\Psi_{++}, \Phi_{+-}$ of the Wiener Hopf equation \eqref{eq.WH3} satisfy
				\begin{alignat*}{3}
				 \numberthis \label{eq.RadlowFunctionalRevised1}	-\Psi_{++}=&  \frac{P_{++}}{K_{++}K_{-+}(\RED{\mathfrak{a}_1},\alpha_2)K_{--}(\RED{\mathfrak{a}_1},\RED{\mathfrak{a}_2})K_{+-}(\alpha_1,\RED{\mathfrak{a}_2})} \\
					& +\frac{1}{K_{++}K_{-+}(\RED{\mathfrak{a}_1},\alpha_2)} \left[\frac{K_{-+}(\RED{\mathfrak{a}_1},\alpha_2)}{K_{+-}} \left[\frac{\Phi_{+-}}{K_{-\circ}}\right]_{+\circ}\right]_{\circ +} \hfill &&\text{for}  \ \boldsymbol{\alpha} \in \D_{++},  \\
				 \numberthis \label{eq.RadlowFunctionalRevisited2} 	0 = & P_{++}\left(
					\frac{1}{K_{--}(\RED{\mathfrak{a}_1},\alpha_2) K_{+-}} - \frac{1}{K_{--}(\RED{\mathfrak{a}_1},\RED{\mathfrak{a}_2})K_{+-}(\alpha_1,\RED{\mathfrak{a}_2})}\right) \\
					& + \left[\frac{K_{-+}(\RED{\mathfrak{a}_1},\alpha_2)}{K_{+-}} \left[\frac{\Phi_{+-}}{K_{-\circ}}\right]_{+\circ}\right]_{\circ -} &&\text{for}  \ \boldsymbol{\alpha} \in \D_{+-}. 
				\end{alignat*}	
\end{theorem} 

\subsection{Significance of Theorem \ref{MainThm}} 
\normalsize
First, note that
the expression for $\Psi_{++}$ in \eqref{eq.RadlowFunctionalRevised1} only differs from Radlow's ansatz given in \cite{Radlow2} by the second term on the equation's RHS. Moreover, it is remarkable that formally  \eqref{eq.RadlowFunctionalRevised1} and \eqref{eq.RadlowFunctionalRevisited2} are almost the same set of equations one obtains for the quarter-plane problem (c.f \cite{AssierAbrahams1}\! eq. 5.12 $\&$ 5.13).\! That is, these equations only differ by the value of the kernel $K=K_{++}K_{+-}K_{--}K_{-+}$ and the sign in front of $\Psi_{++}$ (the latter can be viewed as a notational difference, as discussed in remark \ref{Remark:ComparisonWithQP}). Additionally, if it was somehow possible to invert \eqref{eq.RadlowFunctionalRevisited2} and thus obtain $\Phi_{+-}$ we would obtain $\Psi_{++}$ by \eqref{eq.RadlowFunctionalRevised1}, which by the Wiener-Hopf equation \eqref{eq.1.22} gives $\Phi_{3/4}$ and therefore solves the diffraction problem at hand (by inverse Fourier transform).  

There are several benefits to \eqref{eq.RadlowFunctionalRevised1}. First, it is clear that \eqref{eq.RadlowFunctionalRevised1} indicates Radlow's error as the additional term is missing in his analysis. Second, the constructive procedure given in this article can be helpful in understanding how Radlow's ansatz was obtained and to quantify his error. Indeed, Radlow only stated his solution in \cite{Radlow2} making it difficult to pinpoint where exactly he went wrong.
Additionally, Radlow's ansatz predicts the wrong corner asymptotics as was pointed out by Kraut and Lehmann in \cite{KrautLehmann}. 
Therefore, just as in the quarter-plane case, the correct near field behaviour should be enforced by the additional term in \eqref{eq.RadlowFunctionalRevised1}. This additional term involves the unknown function $\Phi_{+-}$, which should satisfy the \emph{compatibility equation} \eqref{eq.RadlowFunctionalRevisited2}. This equation does not appear in Radlow's work and to our knowledge not in any subsequent work. However, as already pointed out, it is remarkably similar to the compatibility equation found for the quarter-plane diffraction problem in \cite{AssierAbrahams1}. Therefore, we strongly believe it is possible to use \eqref{eq.RadlowFunctionalRevisited2} to test approximations for $\Phi_{+-}$ and thus obtain an approximate solution to $\Psi_{++}$. Indeed, in \cite{AssierAbrahams2} Assier and Abrahams proposed a scheme to accurately approximate $\Phi_{+-}$ for the quarter-plane diffraction problem and we plan to propose a similar method for the penetrable wedge diffraction problem as part of our future work.
Moreover, we do believe that the spectral functions $\Psi_{++}$ and $\Phi_{3/4}$ can be used to obtain far field contributions using the novel `Bridge and Arrow' notation as introduced in \cite{AssierShanin3}, which will also be the basis of future work.


\section{Vanishing imaginary part of the wavenumbers}\label{sec:VanishingImPart}

So far, everything that has been done was under the assumption that $\Im(k_{1,2}) >0$. Let us discuss the limiting procedure $\Im(k_{1,2}) \to 0$. Then the domain of analyticity of $\Psi_{++}$ as discussed in Section \ref{Cha:DomainsOfAnalyticty}, would become $\UHP(0) \times \UHP(0)$. However, due to the incident wave $\phi_{\iin}$, we expect $\Psi_{++}$ to then have polar singularities on the real line at $\alpha_1 = \RED{\mathfrak{a}_1}$ and $\alpha_2= \RED{\mathfrak{a}_2}$ (c.f.\! \cite{AssierAbrahams1, AssierShanin}). Moreover, due to the Kernel, we expect $\Psi_{++}$ to also have branch singularities at $\alpha_1=-k_{1,2},\ \alpha_2=-k_{1,2}$ and polar singularities at some parts of the real circle $\boldsymbol{\alpha}^2 = k^2_2$ (again, c.f.\! \cite{AssierAbrahams1, AssierShanin}). Therefore, when evaluating the physical field
\begin{align}
\psi(\x) = \frac{1}{4 \pi^2} \iint_{\mathbb{R}^2} \Psi_{++}(\boldsymbol{\alpha}) e^{-i  \boldsymbol{\alpha} \cdot \x } d \boldsymbol{\alpha}
\end{align}
we have to indent the `contour' $\mathbb{R}^2$ to $\Gamma \times \Gamma$, see figure \ref{fig:Gamma} (these contours are thoroughly discussed in \cite{AssierAbrahams1}), in order to avoid these singularities\footnote{Here, the choice of incident angle is crucial as this contour is only valid for $\vartheta_0 \in (\pi, 3\pi/2)$.} (note that in the figure the relevant parts of the circle $\boldsymbol{\alpha}^2 =k_2$ are not shown; however, $\Gamma$ also avoids these points). 
By Cauchy's theorem, this does not change the value of $\psi$ and therefore 
\begin{align}
\psi(\x) = \frac{1}{4 \pi^2} \iint_{\Gamma \times \Gamma} \Psi_{++}(\boldsymbol{\alpha}) e^{-i \boldsymbol{\alpha} \cdot \x } d \boldsymbol{\alpha} \label{eq.psiNoImPart}
\end{align}
defines the correct physical field for $\Im(k_{1,2})=0$. Note that to find the singularities of $\Psi_{++}$ which have to be avoided and to make sense of $\Psi_{++}$ in the lower half planes (for vanishing imaginary part), we have to analytically continue $\Psi_{++}$ into a larger domain than that given in Section \ref{Cha:DomainsOfAnalyticty} and unveil its singularities therein. 

\begin{figure}[!h]
\centering
\includegraphics[width=.75\textwidth]{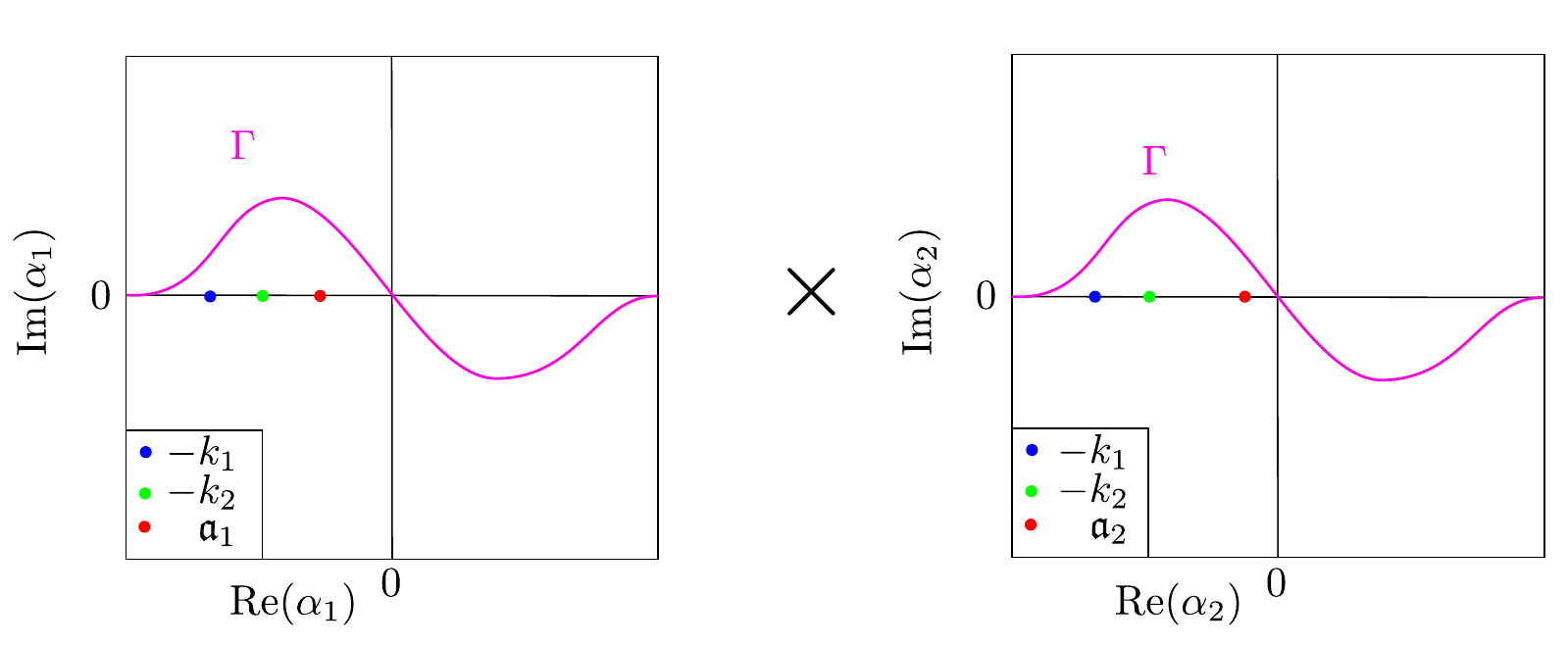}\vspace{-.25cm}
\caption{Contour $\Gamma \times \Gamma$ used in the integral \eqref{eq.psiNoImPart}.}   
\label{fig:Gamma}
\end{figure}

Finally, we mention that $\Phi_{3/4}$ can be dealt with similarly when evaluating $\phi_{\ssc}$ as $\Im(k_{1,2}) \to 0$. That is
\begin{align}
\phi_{\ssc} =\frac{1}{4 \pi^2} \iint_{\Gamma \times \Gamma} \Phi_{3/4}(\boldsymbol{\alpha}) e^{-i  \boldsymbol{\alpha}\cdot \x} d \boldsymbol{\alpha}.
\end{align}


\section{Conclusion}

In this article, we revisited Radlow's double Wiener-Hopf approach to the penetrable wedge diffraction problem. We gave a constructive procedure to obtain his ansatz and hopefully add more clarity to his innovative work. After transforming the physical boundary value problem to two complex dimensional Fourier space, Radlow's Wiener-Hopf equation was recovered, the solution to which directly solves the diffraction problem at hand by inverse Fourier transform. Using the factorisation techniques developed by Assier and Abrahams in \cite{AssierAbrahams1}, the Wiener-Hopf equation \eqref{eq.1.22} was reduced to a coupled system of two functional equations, \eqref{eq.RadlowFunctionalRevised1} and \eqref{eq.RadlowFunctionalRevisited2}, involving two unknowns $\Psi_{++}$ and $\Phi_{+-}$. The first equation involves Radlow's exact ansatz, which gives yet another reason for why his ansatz cannot be the Wiener-Hopf equation's solution (and therefore not solve the diffraction problem at hand). The second equation, the compatibility equation, involves solely the unknown $\Phi_{+-}$. Solving this equation is key to find $\Psi_{++}$, but failing this, we believe it can be used efficiently to find novel approximation schemes for the physical fields.

Finally, it is remarkable how similar the penetrable wedge diffraction problem is to the quarter-plane problem in Fourier space. That is, formally, all occurring relations/equations are almost identical and only differ by $K$'s  structure and $\Psi_{++}$'s sign. Using the novel complex analysis methods developed in  \cite{AssierShanin} and \cite{AssierShanin2} we believe that it is possible to obtain information on the physical field's components by studying the crossing of singularities of $\Psi_{++}$ and $\Phi_{3/4}$. To summarise, this leaves us with the following questions, which we hope to answer in future articles:
\begin{itemize}
\item Applying the methods developed in \cite{AssierAbrahams2}, can we find a new accurate approximation scheme for the penetrable wedge diffraction problem?
\item Using the analytical continuation techniques developed in \cite{AssierShanin}, what more information can we get on $\Psi_{++}$'s and $\Phi_{3/4}$'s domain of analyticity especially regarding their singularity structure?
\item Can the novel Bride and Arrow notation (c.f.\! \cite{AssierShanin2}) be used to obtain far-field asymptotics for $\psi$ and $\phi$?
\end{itemize}


\appendix	

\section{On the derivation of the Wiener-Hopf equation}\label{Appendix: WH-Details}	

Let us show how the Wiener-Hopf equation \eqref{eq.1.22} for the diffraction problem at hand is obtained. Recall the definition of the $1/4$ and $3/4$ Fourier transforms (c.f.\! Definitions \ref{def. 1/4FT} and \ref{def. 3/4FT}) i.e.,
letting $Q_n$ denote the $n$\text{th} quadrant in the plane $\mathbb{R}^2$, we have:
\begin{align}
\mathcal{F}_{3/4}[u](\boldsymbol{\alpha}) & = \iint_{Q_2} u(\x)e^{i \boldsymbol{\alpha}\cdot \x}
d\x + \iint_{Q_3}u(\x)e^{i  \boldsymbol{\alpha}\cdot \x}
d\x  + \iint_{Q_4} u(\x)e^{i \boldsymbol{\alpha}\cdot \x}
d\x \label{eq.3/4Split}, \\[.5em]
\mathcal{F}_{1/4}[u](\boldsymbol{\alpha}) & = \iint_{Q_1}u(\x)e^{i  \boldsymbol{\alpha}\cdot \x}
d\x. 
\end{align}
We apply the operator $\mathcal{F}_{3/4}$ (resp. $\mathcal{F}_{1/4}$) to \eqref{eq:1.1} (resp. \eqref{eq:1.3}).
By Green's second identity we have: 
{\fontsize{10}{0}\selectfont
\begin{align*}
	 \label{eq.1.16}
\numberthis \iint_{Q_1} (\Delta
\psi(\x)) e^{i \boldsymbol{\alpha} \cdot \x} d\x
= & - (\alpha^2_1 + \alpha^2_2) \iint_{Q_1}\psi(\x) e^{i \boldsymbol{\alpha} \cdot \x} d\x \\ 
& - \int_{0}^{\infty} (\partial_{x_1} \psi(0^+,x_2)) e^{i \alpha_2 x_2} dx_2 - \int_{0}^{\infty} (\partial_{x_2} \psi(x_1,0^+)) e^{i \alpha_1 x_1} dx_1  \\
& + i \alpha_1 \int_{0}^{\infty} \psi(0^+,x_2) e^{i \alpha_2 x_2} dx_2 + i \alpha_2 \int_{0}^{\infty} \psi(x_1,0^+) e^{i \alpha_1 x_1} dx_1.  
\end{align*}} 
Similarly, using \eqref{eq.3/4Split} and after a lengthy but straightforward calculation,
we find
\begin{align*}
 \label{eq.1.15}	
\numberthis \mathcal{F}_{3/4}[\Delta \phi_{\ssc}] = & -(\alpha^2_1 + \alpha^2_2) \mathcal{F}_{3/4}(\phi_{\ssc}) + \int_{0}^{\infty}(\partial_{x_1}\phi_{\ssc}(0^-,x_2)) e^{i \alpha_2 x_2} dx_2 \\
& + \int_{0}^{\infty} (\partial_{x_1} \phi_{\ssc}(x_1,0^-)) e^{i \alpha_1 x_1} dx_1 - i \alpha_1 \int_{0}^{\infty} \phi_{\ssc}(0^-,x_2) e^{i \alpha_2 x_2} dx_2  \\
& - i \alpha_2 \int_{0}^{\infty} \phi_{\ssc}(x_1,0^-) e^{i \alpha_1 x_2} dx_1.
\end{align*}
Now	we can use the boundary conditions \eqref{eq:1.4}--\eqref{eq:1.6} to rewrite \eqref{eq.1.16} as 
\begin{align*}
\label{eq.1.17}	
\numberthis \mathcal{F}_{1/4}[\Delta
\psi]
= & - (\alpha^2_1 + \alpha^2_2) \mathcal{F}_{1/4}[\psi]\\ 
& - \left(\int_{0}^{\infty} (\partial_{x_1} \phi_{\ssc}(0^-,x_2)) e^{i \alpha_2 x_2} dx_2 + \int_{0}^{\infty} (\partial_{x_2} \phi_{\ssc}(x_1,0^-)) e^{i \alpha_1 x_1} dx_1 \right) \\
& + i \alpha_1 \int_{0}^{\infty} \phi_{\ssc}(0^-,x_2) e^{i \alpha_2 x_2} dx_2 + i \alpha_2 \int_{0}^{\infty} \phi_{\ssc}(x_1,0^-) e^{i \alpha_1 x_1} dx_1 \\
& - \left(\int_{0}^{\infty} (\partial_{x_1} \phi_{\iin}(0^-,x_2)) e^{i \alpha_2 x_2} dx_2 + \int_{0}^{\infty} (\partial_{x_2} \phi_{\iin}(x_1,0^-)) e^{i \alpha_1 x_1} dx_1 \right) \\
& + i \alpha_1 \int_{0}^{\infty} \phi_{\iin}(0^-,x_2) e^{i \alpha_2 x_2} dx_2 + i \alpha_2 \int_{0}^{\infty} \phi_{\iin}(x_1,0^-) e^{i \alpha_1 x_1} dx_1. 
\end{align*}	
But $\phi_{\iin}= \exp(-i(a_1x_1 + a_2x_2))$ so, since $-\mathrm{Im}(\RED{\mathfrak{a}_{1,2}}) > 0$, we calculate:
\begin{align*}
\numberthis \int_{0}^{\infty} (\partial_{x_1} \phi_{\iin}(0^-,x_2)) e^{i \alpha_2 x_2} dx_2 &= -ia_1 \int_{0}^{\infty} e^{i (-\RED{\mathfrak{a}_2} +\alpha_2) x_2} dx_2 \\
& = \frac{\RED{\mathfrak{a}_1}}{\alpha_2 -\RED{\mathfrak{a}_2}} . 
\end{align*}
Similarly, we compute the other terms in \eqref{eq.1.17} involving $\phi_{\iin}$, and obtain
\begin{align*}
\label{eq:2.22}	
\numberthis  \mathcal{F}_{3/4}[\Delta \phi_{\ssc}] + \mathcal{F}_{1/4}[\Delta \psi ]
=& -(\alpha^2_1 + \alpha^2_2)\mathcal{F}_{3/4}[\phi_{\ssc}] -(\alpha^2_1 + \alpha^2_2)\mathcal{F}_{1/4}[\psi] \\
& -\frac{\RED{\mathfrak{a}_1}}{\alpha_2 - \RED{\mathfrak{a}_2}} - \frac{\RED{\mathfrak{a}_2}}{\alpha_1 -\RED{\mathfrak{a}_1}}
- \frac{\alpha_1}{\alpha_2 -\RED{\mathfrak{a}_2}} - \frac{\alpha_2}{\alpha_1 -\RED{\mathfrak{a}_1}}. 
\end{align*}	 
Thus
\begin{align*}
\label{eq.1.19} 
\numberthis 0 =& \mathcal{F}_{3/4}[\Delta \phi_{\ssc} +k^2_1 \phi_{\ssc}] + \mathcal{F}_{1/4}[\Delta \psi +k^2_2 \psi]  \\
= & (k^2_1-\alpha^2_1 - \alpha^2_2)\left(\mathcal{F}_{3/4}[\phi_{\ssc}] + \frac{1}{( \alpha_1 -\RED{\mathfrak{a}_1})(\alpha_2 -\RED{\mathfrak{a}_2})}\right) +(k^2_2 - \alpha^2_1 - \alpha^2_2)\mathcal{F}_{1/4}[\psi]. 
\end{align*}	
which is equivalent to the Wiener-Hopf equation \eqref{eq.1.22}. \\ 

\subsection{On the importance of $\lambda =1$}\label{sec.ImportanceLambda=1}
Recall that throughout this article, we assume $\lambda=1$ for the contrast parameter $\lambda$ (c.f.\! Section \ref{sec:ProblemFormulation}). If this was not the case, i.e. for a general $\lambda$, the corresponding boundary conditions for the normal derivative would, instead of \eqref{eq:1.5} and \eqref{eq:1.6}, read 
\begin{align}
& \frac{1}{\lambda} \partial_{x_1} \phi(0^-, x_2>0) = \partial_{x_1} \psi(0^+, x_2>0)  \\
& \frac{1}{\lambda}  \partial_{x_2} \phi(x_1>0, 0^-) =  \partial_{x_2} \psi(x_1>0, 0^+). 
\end{align}
But using these boundary conditions, and repeating the preceding procedure, we would instead of \eqref{eq.1.17} find 
\begin{align*}
\label{eq.LambdaProblem}
\numberthis \mathcal{F}_{1/4}[\Delta
\psi]
= & - (\alpha^2_1 + \alpha^2_2) \mathcal{F}_{1/4}[\psi]\\ 
& - \frac{1}{\lambda} \left(\int_{0}^{\infty} (\partial_{x_1} \phi_{\ssc}(0^-,x_2)) e^{i \alpha_2 x_2} dx_2 + \int_{0}^{\infty} (\partial_{x_2} \phi_{\ssc}(x_1,0^-)) e^{i \alpha_1 x_1} dx_1 \right) \\
& + i \alpha_1 \int_{0}^{\infty} \phi_{\ssc}(0^-,x_2) e^{i \alpha_2 x_2} dx_2 + i \alpha_2 \int_{0}^{\infty} \phi_{\ssc}(x_1,0^-) e^{i \alpha_1 x_1} dx_1 \\
& - \frac{1}{\lambda} \left(\int_{0}^{\infty} (\partial_{x_1} \phi_{\iin}(0^-,x_2)) e^{i \alpha_2 x_2} dx_2 + \int_{0}^{\infty} (\partial_{x_2} \phi_{\iin}(x_1,0^-)) e^{i \alpha_1 x_1} dx_1 \right) \\
& + i \alpha_1 \int_{0}^{\infty} \phi_{\iin}(0^-,x_2) e^{i \alpha_2 x_2} dx_2 + i \alpha_2 \int_{0}^{\infty} \phi_{\iin}(x_1,0^-) e^{i \alpha_1 x_1} dx_1  
\end{align*}
whilst equation \eqref{eq.1.15} obtained for $\mathcal{F}_{3/4}[\Delta \phi_{\ssc}]$ remains the same. But then the boundary terms on the RHS of \eqref{eq.LambdaProblem} not including the field's normal derivative do not cancel with the corresponding boundary terms in \eqref{eq.1.15} when considering $\mathcal{F}_{3/4}[\Delta \phi_{\ssc}] + \mathcal{F}_{1/4}[\Delta \psi]$ and therefore we would not obtain the Wiener-Hopf equation \eqref{eq.1.22}.

\section{Asymptotic behaviour of spectral functions}\label{sec.LiouvilleApplication}

Let us investigate the far-field behaviour of $\Psi_{++}$. For this, we need to invoke the following essential theorem: 

\begin{theorem}[Abelian Theorem]\label{thm:Abelian}
Suppose that two real-valued functions $f(r)$, $g(r)$, defined for $r>0$ are continuous in some interval $0<r<R$ where $g(r)\neq0$. 
Assume that all following transformations are well defined. Then, if 
\begin{align*}
	f(r) \sim g(r), \ \text{as} \ r \to 0 
\end{align*} 
we have 
\begin{align*}
	\int_{0}^{\infty} f(r)e^{irz} dr \sim \int_{0}^{\infty} g(r)e^{irz} dr, \ \text{as}
	\ |z| \to \infty \ \text{within} \ \UHP(0).
\end{align*}	

\end{theorem}
The proof can be found in \cite{Doetsch} Theorem 33.1\RED{,} for instance. Now, by the edge condition \eqref{eq.2.57}, we know 
\begin{align*}
\psi(\x) \sim B \ \text{as}, \ |\x| \to 0
\end{align*}
for a suitable constant $B$. Moreover, we can use the following trick used by Assier and Abrahams in \cite{AssierAbrahams2} and see that for any $\varepsilon >0$ we have 
\begin{align}
\psi(\x) \sim B e^{-\varepsilon x_1 - \varepsilon x_2} \ \text{as} \ |\x| \to 0.\label{eq.4.5}
\end{align}
Choose $\varepsilon = 2 \delta$ so $\alpha_{1,2}+i\varepsilon$ has strictly positive imaginary part.
Observe that \eqref{eq.4.5} implies
\begin{align}
\psi(0,x_2) \sim B e^{-\varepsilon x_2}, \ \text{as} \ x_2 \to 0 \label{eq.4.6}
\end{align}
which, in particular, gives 
\begin{align}
|\psi(0,x_2)| \sim |B| e^{-\varepsilon x_2}, \ \text{as} \ x_2 \to 0. \label{eq.4.7}
\end{align} 	
Finally, \eqref{eq.4.5} yields 
\begin{align}
\psi(\x) \sim B \psi(0,x_2) e^{-\varepsilon x_1}, \  \text{as} \ x_1 \to 0. \label{eq.4.8}
\end{align}	
Then, invoking the Abelian theorem, we first obtain using \eqref{eq.4.8} 
\begin{align}
\int_{0}^{\infty} \psi(\x)e^{i\alpha_1 x_1} dx_1 \sim B \psi(0,x_2) \int_{0}^{\infty} e^{i(\alpha_1+i \varepsilon)x_1} dx_1 =   \frac{B \psi(0,x_2)}{\alpha_1 + i \varepsilon}, \ \text{as} \ |\alpha_1| \to \infty. \label{eq.PsiAsympt1}
\end{align}
Now, due to  \eqref{eq.4.7}, invoking the Abelian theorem once again, we find
\begin{align}
\int_{0}^{\infty}|\psi(0,x_2)| e^{i \alpha_2 x_2} dx_2 \sim |B| \int_0^{\infty} e^{i(\alpha_2+i \varepsilon)x_2} dx_2= |B| \frac{1}{\alpha_2 + i \varepsilon}, \ \text{as} \ |\alpha_2| \to \infty, \label{eq.PsiAsympt2}
\end{align}
and therefore:
\begin{lemma}\label{lem:Psi++Asympt}
For fixed $\alpha^*_2$ (resp. fixed $\alpha^*_1$) in $\UHP$ we have 
\begin{align}
	&	\Psi_{++}(\alpha_1,\alpha^*_2) = \mathcal{O}(1/|\alpha_1|), \ \text{as} \ |\alpha_1| \to \infty \ \text{in} \ \UHP  \label{Psi++Behaviour1} \\
	&	\Psi_{++}(\alpha^*_1,\alpha_2) = \mathcal{O}(1/|\alpha_2|), \ \text{as} \ |\alpha_2| \to \infty \ \text{in} \ \UHP \label{Psi++Behaviour1.2}
\end{align}
and, if neither variable is fixed, 
\begin{align}
	\Psi_{++}(\alpha_1, \alpha_2) = & \mathcal{O}(1/|\alpha_1||\alpha_2|) \ \text{as} \ 
	|\alpha_1| \to \infty, \ |\alpha_2| \to \infty \ \text{in} \ \UHP.   \label{Psi++Behaviour2}
\end{align}	
\end{lemma}
\begin{proof}
We obtain \eqref{Psi++Behaviour1} from \eqref{eq.PsiAsympt1} and \eqref{Psi++Behaviour1.2} from \eqref{eq.PsiAsympt2}. To get \eqref{Psi++Behaviour2} combine \eqref{eq.PsiAsympt1} and \eqref{eq.PsiAsympt2}.
\end{proof}
Similarly, estimates for $\Phi_{-+}, \ \Phi_{--},$ and $\Phi_{+-}$ are obtained:
\begin{lemma}\label{lem:PhiAsympt.}
The functions $\Phi_{-+}, \Phi_{--}$, and $\Phi_{+-}$ satisfy the \RED{decay} estimates \eqref{Psi++Behaviour1}--\eqref{Psi++Behaviour2} as $|\alpha_{1,2}| \to \infty$ in these function's respective domains.
\end{lemma}

\section{On the application of Liouville's theorem}

In order to apply the results of Lemma \ref{lem:Psi++Asympt} and \ref{lem:PhiAsympt.} to the functions $E_1$ and $E_2$ defined in Section \ref{sec:WHSystem},
we need to establish a link between the decay of a function $f(z)$ and \RED{the functions} $f_-(z)$ and $f_+(z)$ \RED{defined by the sum split $f(z)=f_+(z) +f_-(z)$ (c.f. Theorem \ref{thm. Cauchy Sum-split})}. 

\begin{theorem}[\RED{Decay} estimates for sum-split]\label{thm:GrowthEstimateSumSplit} \ Let $f(z)$ be a function analytic on some strip, and consider its sum-split $f(z) = f_+(z) + f_{-}(z)$. 
\begin{enumerate}
	\item If $f(z)=\mathcal{O}\left(1 /|z|^{\lambda}\right)$ as $|z| \rightarrow \infty$ within the strip, with $\lambda>1,$ then $f_{\pm}(z)$ are decaying at least like $1 /|z|$ as $|z| \rightarrow \infty$ within their respective half-planes.
	\item If $f(z)=\mathcal{O}(1 /|z|)$ as $|z| \rightarrow \infty$ within the strip, then $f_{\pm}(z)$ are decaying at least like $\ln |z| /|z|$ as $|z| \rightarrow \infty$ within their respective half-planes. 
	\item If  $f(z)=\mathcal{O}\left(1 /|z|^{\lambda}\right)$ \text { as }$|z| \rightarrow \infty$ within the strip, with $0<\lambda<1$, \text { then } $f_{\pm}(z)$ \text { are decaying } \text { at least like } $1 /|z|^{\lambda}$ \text { as }$|z| \rightarrow \infty$  within their respective half-planes.
\end{enumerate}
\end{theorem}

For the proof see \cite{Woolcock}. However, Theorem \ref{thm:GrowthEstimateSumSplit} is a summary of the results given in \cite{Woolcock}, applicable to the problem at hand. The summary is taken from \cite{AssierAbrahams1} (c.f.\! Lemma B.1 therein). 

\subsection{Application in the $\alpha_1$ plane}

\begin{lemma}\label{Lemma:LiouvilleAlpha1}
The function $E_1$ given by 
\begin{align*}
	E_1(\alpha_1,\alpha_2) = \begin{cases}
		-K_{+\circ} \Psi_{++} -\left[\frac{\Phi_{+-}}{K_{- \circ}}\right]_{+\circ}\!\!  -\frac{P_{++}}{K_{-\circ}(\RED{\mathfrak{a}_1},\alpha_2)},&\text{if} \ \boldsymbol{\alpha} \in \D_{+\circ}, \\
		\frac{\Phi_{-\circ}}{K_{- \circ}} + \left[\frac{\Phi_{+-}}{K_{- \circ}}\right]_{-\circ}\!\!+
		P_{++}\left(\frac{1}{K_{-\circ}} - \frac{1}{K_{-\circ}(\RED{\mathfrak{a}_1},\alpha_2)} \right),\!\!& \text{if} 
		\   \boldsymbol{\alpha} \in \D_{-\circ}.
	\end{cases}
\end{align*}
vanishes i.e $E_1\equiv 0$.
\end{lemma} 

\medskip 
\begin{proof}
Let us fix some \RED{$\alpha_2=\alpha^*_2 \in \S$.} \RED{By Lemma \ref{lemma:KpmcircDomain} we know $K_{+\circ} \to 1$, as  $|\alpha_1| \to \infty$ in UHP, and by definition of $P_{++}$ (c.f. \eqref{eq.KPDef})} it is clear that 	
\begin{align*}
	P_{++} \to 0, \ \text{as} \ |\alpha_1| \to \infty, \ \text{in} \ \mathrm{UHP}.
\end{align*}
But due to Lemma\RED{s} \ref{lem:Psi++Asympt}, \ref{lem:PhiAsympt.} and Theorem \ref{thm:GrowthEstimateSumSplit} we know that
$[\Phi_{+-}/K_{-\circ}]_{\pm\circ}$ decays at least like $\ln|\alpha_1|/|\alpha_1|$ as $|\alpha_1| \to \infty$ in $\UHP$ (resp. $\LHP$). So we know that 
$E \to 0$ as $|\alpha_1| \to \infty$ in $\UHP$. Similarly, we find $E \to 0$ as $|\alpha_1| \to \infty$ in $\LHP$ and therefore, \RED{since $\UHP\cap \LHP = \S$ is not empty (c.f. Section \ref{sec:SetNotations}),} by Liouville's theorem applied in the $\alpha_1$ plane, $E_1 \equiv 0$. 
\end{proof}
\subsection{Application in the $\alpha_2$ plane}

\begin{lemma}\label{Lemma:LiouvilleAlpha2}
The function $E_2(\alpha_1,\alpha_2)$ given by 
{\fontsize{7.85pt}{0pt}\selectfont
	\begin{align}
		\nonumber
		E_{2}=  \begin{cases}
			-\Psi_{++}K_{++}K_{-+}(\RED{\mathfrak{a}_1},\alpha_2)  - \frac{P_{++}}{K_{--}(\RED{\mathfrak{a}_1},\RED{\mathfrak{a}_2})K_{+-}(\alpha_1,\RED{\mathfrak{a}_2})}
			-  \left[\frac{K_{-+}(\RED{\mathfrak{a}_1},\alpha_2)}{K_{+-}} \left[\frac{\Phi_{+-}}{K_{-\circ}}\right]_{+\circ}\right]_{\circ +}\!\!, \ \boldsymbol{\alpha} \in \D_{++} \\
			P_{++}\left(
			\frac{1}{K_{--}(\RED{\mathfrak{a}_1},\alpha_2) K_{+-}} \right. - \left. \frac{1}{K_{--}(\RED{\mathfrak{a}_1},\RED{\mathfrak{a}_2})K_{+-}(\alpha_1,\RED{\mathfrak{a}_2})}\right)
			+  \left[\frac{K_{-+}(\RED{\mathfrak{a}_1},\alpha_2)}{K_{+-}} \left[\frac{\Phi_{+-}}{K_{-\circ}}\right]_{+\circ}\right]_{\circ -}\!\!\!\!, \  \boldsymbol{\alpha} \in \D_{+-}. 
		\end{cases}
	\end{align}
}
\normalsize
vanishes i.e $E_2 \equiv 0$.
\end{lemma}

\begin{proof} Applying Theorem \ref{thm:GrowthEstimateSumSplit} to $\mylog(K_{\pm \circ})$,  we find, after applying $\exp$, that $K_{--}$, $\ K_{-+},$ $\ K_{--},$ $\ K_{+-}$ all go to $1$ as $|\alpha_2| \to \infty$ in $\UHP$ and $\LHP$ respectively. Moreover, $P_{++} \to 0$ as $|\alpha_2| \to \infty$ in $\UHP$ and $\LHP$ respectively. Therefore, applying Theorem \ref{thm:GrowthEstimateSumSplit} to $\Psi_{++}$, $\Phi_{+-}$, $\Phi_{--}$, and $\Phi_{-+}$ (using the estimates given in Lemma\RED{s} \ref{lem:Psi++Asympt} and \ref{lem:PhiAsympt.}) we find that all terms except possibly the bracket terms in $E_2$ vanish as $|\alpha_2| \to \infty$. 
But we can use \eqref{eq.3.27} to directly obtain estimates for the behaviour of $\frac{K_{-+}(\RED{\mathfrak{a}_1},\alpha_2)}{K_{+-}}\left[\frac{\Phi_{+-}}{K_{-\circ}}\right]_{+\circ}$ as $|\alpha_2| \to \infty$ in $\S$ and thereafter apply Theorem \ref{thm:GrowthEstimateSumSplit} (finding that the bracket terms vanishes as $|\alpha_2| \to \infty$ in $\UHP$ and $\LHP$ respectively).  See \cite{AssierAbrahams1} for a more detailed discussion.	
\end{proof}

\bibliographystyle{plain}
\bibliography{bibliography}

\begin{thebibliography}{10}

\bibitem{AssierAbrahams2}
R.~C. Assier and I.~D. Abrahams.
\newblock On the asymptotic properties of a canonical diffraction integral.
\newblock {\em Proc. R. Soc. A}, 476:20200150, 2020.

\bibitem{AssierAbrahams1}
R.~C. Assier and I.~D. Abrahams.
\newblock A surprising observation on the quarter-plane diffraction problem.
\newblock {\em SIAM J. Appl. Math}, 81(1):60--90, 2021.

\bibitem{AssierShanin}
R.~C. Assier and A.~V. Shanin.
\newblock Diffraction by a quarter-plane. {A}nalytical continuation of spectral
  functions.
\newblock {\em Q . Jl Mech. Appl. Math}, 72(1), 2019.

\bibitem{AssierShanin3}
R.~C. Assier and A.~V. Shanin.
\newblock Analytical continuation of two-dimensional wave fields.
\newblock {\em Proc. Roy. Soc. A}, 477(2020081), 2021.

\bibitem{AssierShanin2}
R.~C. Assier and A.~V. Shanin.
\newblock Vertex {G}reen’s functions of a quarter-plane. links between the
  functional equation, additive crossing and {L}am\'e functions.
\newblock {\em Q.J. Mech. Appl. Math.}, 74(3), 2021.

\bibitem{Babich}
V.~M. Babich, M.~A. Lyalinov, and V.~E. Grikurov.
\newblock Diffraction theory: The {S}ommerfeld-{M}alyuzhinets technique (alpha
  science series on wave phenomena).
\newblock {\em Oxford: Alpha Science}, 2007.

\bibitem{BabichMokeeva}
V.~M. Babich and N.~V. Mokeeva.
\newblock Scattering of the plane wave by a transparent wedge.
\newblock {\em J. Math. Sci.}, 155(3):335–342, 2008.

\bibitem{BabichEtAl.}
V.~M. Babich, N.~V. Mokeeva, and B.~A. Samokish.
\newblock The problem of scattering of a plane wave by a transparent wedge: A
  computational approach.
\newblock {\em Commun.Technol. Electron.}, 57(9):993–1000, 2012.

\bibitem{Baran}
A.~J. Baran.
\newblock {\em Light Scattering by Irregular Particles in the Earth’s
  Atmosphere}, volume~8 of {\em Light Scattering Reviews}.
\newblock Berlin, Heidelberg: Springer, 2013.

\bibitem{BelinskiyEtAl1973}
B.P. Belinskiy, D.P. Kouzov, and V.D. Cheltsova.
\newblock On acoustic wave diﬀraction by plates connected at a right angle.
\newblock {\em J. of Applied Math. and Mechanics}, 37(2):273--281, 1973.

\bibitem{Budaev1}
B.~V. Budaev and D.~B. Bogy.
\newblock Rayleigh wave scattering by two adhering elastic wedges.
\newblock {\em Proc. R. Soc. A Math. Phys. Eng. Sci.}, 454(1979):2949–2996,
  1998.

\bibitem{Budaev2}
B.~V. Budaev and D.~B. Bogy.
\newblock Rigorous solutions of acoustic wave diffraction by penetrable wedges.
\newblock {\em J. Acoust. Soc. Am.}, 105(1):74–83, 1999.

\bibitem{CroisilleLebeau}
J.~P. Croisille and G.~Lebeau.
\newblock {\em Diffraction by an Immersed Elastic Wedge}.
\newblock Springer-Verlag, Berlin, Heidelberg, 1999.

\bibitem{DanieleLombardi}
V.~Daniele and G.~Lombardi.
\newblock The {W}iener-{H}opf solution of the isotropic penetrable wedge
  problem: Diffraction and total field.
\newblock {\em IEEE Transactions on Antennas and Propagation}, 59, 2011.

\bibitem{Doetsch}
G.~Doetsch.
\newblock {\em Introduction to the Theory and Application of the Laplace
  Transformation}.
\newblock Springer-Verlag Berlin Heidelberg New York, 1974.

\bibitem{GrothEtAl.1}
S.~P. Groth, D.~P. Hewett, and S.~Langdon.
\newblock Hybrid numerical-asymptotic approximation for high-frequency
  scattering by penetrable convex polygons.
\newblock {\em IMA J. Appl. Math.}, 80, 2015.

\bibitem{GrothEtAl.2}
S.~P. Groth, D.~P. Hewett, and S.~Langdon.
\newblock A high frequency boundary element method for scattering by penetrable
  convex polygons.
\newblock {\em Wave Motion}, 78, 2018.

\bibitem{Jones1964}
D.~S. Jones.
\newblock {\em The Theory of Electromagnetism}.
\newblock Elsevier Ltd., 1964.

\bibitem{Keller}
J.~B. Keller.
\newblock Geometrical theory of diffraction.
\newblock {\em Journal of the Optical Society of America}, 52, 1962.

\bibitem{KinslerEtAl}
L.~E. Kinsler, A.~R. Frey, A.~B. Coppens, and J.~V. Sanders.
\newblock {\em Fundamentals of Acoustics Fourth Edition}.
\newblock John Wiley and Sons, Inc., 1999.

\bibitem{KontorovichLebedev}
M.~J. Kontorovich and N.~N. Lebedev.
\newblock On a method of solution of some problems of the diffraction theory.
\newblock {\em J. Phys. (Academy Sci. U.S.S.R.)}, 1, 1939.

\bibitem{KrautLehmann}
E.~A. Kraut and G.~W. Lehmann.
\newblock Diffraction of electromagnetic waves by a right-angle dielectric
  wedge.
\newblock {\em Journal of Mathematical Physics}, 10, 1969.

\bibitem{LawrieAbrahams}
J.~B. Lawrie and I.~D. Abrahams.
\newblock A brief historical perspective of the {W}iener-{H}opf technique.
\newblock {\em J. Eng. Math.}, 59, 2007.

\bibitem{Lyalinov}
M.~A. Lyalinov.
\newblock Diffraction by a highly contrast transparent wedge.
\newblock {\em J. Phys. A. Math. Gen.}, 32, 1999.

\bibitem{Malyuzhinets}
G.~D. Malyuzhinets.
\newblock Excitation, reflection and emission of surface waves from a wedge
  with given face impedances.
\newblock {\em Sov. Phys. Dokl.}, 3, 1958.

\bibitem{Meister}
E.~Meister.
\newblock Some solved and unsolved canonical problems of diffraction theory.
\newblock {\em Lect. Notes Math.}, 1285, 1987.

\bibitem{Nethercote1}
M.~A. Nethercote, R.~C. Assier, and I.~D. Abrahams.
\newblock Analytical methods for perfect wedge diffraction: a review.
\newblock {\em Wave Motion}, 93, 2020.

\bibitem{Nethercote2}
M.~A. Nethercote, R.~C. Assier, and I.~D. Abrahams.
\newblock High-contrast approximation for penetrable wedge diffraction.
\newblock {\em IMA J. Appl. Math.}, 85(3):421--466, 2020.

\bibitem{Noble}
B.~Noble.
\newblock {\em Methods Based on the {W}iener-{H}opf Technique}.
\newblock Pergamon Press London, Neq York, Paris, Los Angeles, 1958.

\bibitem{Radlow1}
J.~Radlow.
\newblock Diffraction by a quarter-plane.
\newblock {\em Arch. Ration. Mech. Anal.}, 8, 1961.

\bibitem{Radlow2}
J.~Radlow.
\newblock Diffraction by a right-angled dielectric wedge.
\newblock {\em ht. J. Engng. Sei.}, 2, 1964.

\bibitem{Rawlins1977}
A.~D. Rawlins.
\newblock Diffraction by a dielectric wedge.
\newblock {\em J. Inst. Maths Applics}, 18:231--279, 1977.

\bibitem{Rawlins}
A.~D. Rawlins.
\newblock Diffraction by, or diffusion into, a penetrable wedge.
\newblock {\em Proc. R. Soc. A Math. Phys. Eng. Sci.}, 455, 1999.

\bibitem{SmithEtAl.}
H.~R. Smith, A.~Webb, P.~Connolly, and A.~J. Baran.
\newblock Cloud chamber laboratory investigations into the scattering
  properties of hollow ice particles.
\newblock {\em J. Quant. Spectrosc. Radiat. Transf.}, 157:106–118, 2015.

\bibitem{Sommerfeld1}
A.~Sommerfeld.
\newblock Mathematische {T}heorie der {D}iffraction.
\newblock {\em Mathematische Annalen}, 47(2-3):317–374, 1896.

\bibitem{Sommerfeld3}
A.~Sommerfeld.
\newblock Theoretisches \"uber die {B}eugung der {R}\"ontgenstrahlen.
\newblock {\em Zeitschrift f\"ur Mathematik und Physik}, 46, 1901.

\bibitem{Sommerfeld2}
A.~Sommerfeld, R.~J. Nagem, M.~Zampolli, and G.~Sandri.
\newblock {\em Mathematical Theory of Diffraction}.
\newblock (Progress in Mathematical Physics). Birkh\"auser, 2004.

\bibitem{Wegert}
E.~Wegert.
\newblock {\em Visual Complex Functions \ an introduction with phase
  portraits}.
\newblock Birkh\"auser Verlag, 2012.

\bibitem{Woolcock}
W.~S. Woolcock.
\newblock Asymptotic behavior of {S}tieltjes transforms.
\newblock {\em I. Journal of Mathematical Physics}, 8(6):1270–1275, 1967.

\end{thebibliography}

\end{document}